%% file: event_studies_paper.tex
\documentclass[11pt,english]{article}
\usepackage{amsmath}
\usepackage{amsthm}
\usepackage{mathptmx}
\usepackage{newtxmath}

\usepackage[T1]{fontenc}
\usepackage[latin9]{inputenc}
\usepackage{array}
\usepackage{float}
\usepackage{rotfloat}
\usepackage{url}
\usepackage{graphicx}
\usepackage{setspace}
\usepackage[authoryear]{natbib}
\makeatletter

\providecommand{\tabularnewline}{\\}

  \theoremstyle{definition}
  \newtheorem{defn}{\protect\definitionname}
  \theoremstyle{plain}
  \newtheorem{assumption}{\protect\assumptionname}
  \theoremstyle{plain}
  \newtheorem{prop}{\protect\propositionname}
  \theoremstyle{remark}
  \newtheorem{rem}{\protect\remarkname}
  \theoremstyle{plain}
  \newtheorem{lem}{\protect\lemmaname}

\usepackage{babel}

\usepackage{setspace}

\usepackage{footmisc}

\usepackage[letterpaper, portrait, margin=1in]{geometry}

\usepackage{mathtools}
\usepackage{hyperref}

    \usepackage[position=top]{subfig}
\usepackage[titletoc,title]{appendix}

\usepackage{enumitem}

\usepackage{appendix}

\date{September 22, 2020}

\@ifundefined{showcaptionsetup}{}{%
 \PassOptionsToPackage{caption=false}{subfig}}
\usepackage{subfig}
\makeatother

\usepackage{babel}
  \providecommand{\assumptionname}{Assumption}
  \providecommand{\definitionname}{Definition}
  \providecommand{\lemmaname}{Lemma}
  \providecommand{\propositionname}{Proposition}
  \providecommand{\remarkname}{Remark}

\begin{document}

\title{Estimating Dynamic Treatment Effects in Event Studies with Heterogeneous
Treatment Effects\thanks{{\footnotesize{}We are grateful to Isaiah Andrews, Amy Finkelstein,
Anna Mikusheva, and Heidi Williams for their guidance and support.
We thank Alberto Abadie, Jonathan Cohen, Nathan Hendren, Peter Hull,
Guido Imbens, Yunan Ji, Sylvia Klosin, Kevin Kainan Li, Paichen Li,
Therese A. McCarty, Whitney Newey, James Poterba, Pedro H. C. Sant'Anna,
Gergely Ujhelyi and Helen Willis for helpful discussions. This research
was supported by the National Institute on Aging, Grant Number T32-AG000186.
A preliminary draft of this paper was circulated on April 16, 2018.
Replication code is available at \protect\url{http://economics.mit.edu/grad/lsun20/}.
The companion Stata package \texttt{eventstudyweights} is available
from the SSC repository.}}}

\author{Liyang Sun\thanks{Department of Economics, MIT, 77 Massachusetts Avenue, Cambridge,
MA 02139. Corresponding author; lsun20@mit.edu.} \ and Sarah Abraham\thanks{Cornerstone Research, 699 Boylston St, Boston, MA 02116. The views
expressed herein are solely those of the author, who is responsible
for the content, and do not necessarily represent the views of Cornerstone
Research. }}

\maketitle
\thispagestyle{empty}
\begin{abstract}
To estimate the dynamic effects of an absorbing treatment, researchers
often use two-way fixed effects regressions that include leads and
lags of the treatment. We show that in settings with variation in
treatment timing across units, the coefficient on a given lead or
lag can be contaminated by effects from other periods, and apparent
pretrends can arise solely from treatment effects heterogeneity. We
propose an alternative estimator that is free of contamination, and
illustrate the relative shortcomings of two-way fixed effects regressions
with leads and lags through an empirical application.
\end{abstract}
 \newpage{}

\section{Introduction}

Rich panel data has fueled a growing literature estimating treatment
effects with two-way fixed effects regressions. This body of applied
work has prompted a corresponding econometrics literature investigating
the assumptions required for these regressions to yield causally interpretable
estimates. For example, \citet{athey_imbens_staggered_adoption},
\citet{kirill}, \citet{callaway_santanna_ssrn2018}, \citet{double_fe}
and \citet{goodman-bacon_difference--differences_2018} interpret
the coefficient on the treatment status when there is treatment effects
heterogeneity and variation in treatment timing. Researchers are often
also interested in dynamic treatment effects, which they estimate
by the coefficients $\mu_{\ell}$ associated with indicators for being
$\ell$ periods relative to the treatment, in a specification that
resembles the following:
\begin{equation}
Y_{i,t}=\alpha_{i}+\lambda_{t}+\sum_{\ell}\mu_{\ell}\mathbf{1}\{t-E_{i}=\ell\}+\upsilon_{i,t}.\label{eq:TWFE intro}
\end{equation}
Here $Y_{i,t}$ is the outcome of interest for unit $i$ at time $t$,
$E_{i}$ is the time when unit $i$ initially receives the binary
absorbing treatment, and $\alpha_{i}$ and $\lambda_{t}$ are the
unit and time fixed effects. Units are categorized into different
cohorts based on their initial treatment timing. The relative times
$\ell=t-E_{i}$ included in (\ref{eq:TWFE intro}) cover most of the
possible relative periods, but may still exclude some periods. 

The first goal of this paper is to uncover potential pitfalls associated
with using the estimates of the relative period coefficients $\mu_{\ell}$
as ``reasonable'' measures of dynamic treatment effects. We decompose
$\mu_{\ell}$ to show it can be expressed as a linear combination
of cohort-specific effects from both its own relative period $\ell$
and other relative periods; unless strong assumptions regarding treatment
effects homogeneity hold, the terms that include treatment effects
from other relative periods will not cancel out and will contaminate
the estimate of $\mu_{\ell}.$ Importantly, this demonstrates that
the widespread practice of using estimates of treatment leads in (\ref{eq:TWFE intro})
as a way of testing for parallel pretrends is problematic. \citet{roth_pretest_2019},
in his survey of the applied literature, notes that checking whether
$\mu_{\ell}=0$ for $\ell$ leads of treatment is a common test for
pretrends. Our decomposition result implies that such a test would
be invalid because the estimate of $\mu_{\ell}$ is affected by both
pretrends and treatment effects heterogeneity, thus any test of $\mu_{\ell}=0$
cannot accept or reject the existence of pretrends without further
assumptions on treatment effects. 

We show how to calculate the weights underlying the linear combination
of treatment effects in $\mu_{\ell}$ using an auxiliary regression.
This auxiliary regression depends only on the distribution of cohorts
and the relative time indicators included in (\ref{eq:TWFE intro}).
Examining the weights allows researchers to  gauge how large the amount
of treatment effects heterogeneity needs to be for $\mu_{\ell}$ to
be contaminated by treatment effects from other relative periods.
Our publicly-available Stata package {\footnotesize{}\texttt{eventstudyweights}}
automates the estimation of these weights using the panel dataset
underlying any given specification of (\ref{eq:TWFE intro}).

The second goal of this paper is to propose an alternative regression-based
method that is more robust to treatment effects heterogeneity than
regression (\ref{eq:TWFE intro}). For dynamic treatment effects,
researchers are usually interested in estimating some average of treatment
effects from $\ell$ periods relative to the treatment. Our alternative
method estimates the shares of cohort as weights. These weights are
more interpretable than the weights underlying regression (\ref{eq:TWFE intro})
in the presence of treatment effects heterogeneity, and the resulting
weighted average of treatment effects extends beyond a convex combination
of treatment effects \citep{sloczynski_general_2018}. As discussed
in Section \ref{subsec:Difference-in-differences-estima}, using the
procedures of \citet{callaway_santanna_ssrn2018}, our alternative
method can also accommodate covariates.

We illustrate both our decomposition results and our alternative method
via an empirical application, estimating the dynamic effects of a
hospitalization. We follow \citet{dobkin_finkelstein_kluender_notowidigdo_aer2018}
in using the publicly-available dataset, Health and Retirement Study
(HRS), to first estimate two-way fixed effects regressions. We then
illustrate our alternative estimation method with this example. Among
the outcomes studied by \citet{dobkin_finkelstein_kluender_notowidigdo_aer2018},
we focus on out-of-pocket medical spending and labor earnings. Our
alternative method yields similar big-picture findings as the original
paper that uses two-way fixed effects regressions: the earnings decline
due to hospitalization is substantial compared to the transitory out-of-pocket
spending increase. However, the two-way fixed effects estimates sometimes
fall outside the convex hull of the underlying effects. In contrast,
estimates using our alternative method, by construction, are guaranteed
to be easy-to-interpret because they are weighted averages of the
underlying effects, with weights corresponding to cohort shares.

The rest of the paper is organized as follows. In the next subsection,
we review the theoretical literature. Section \ref{sec:Event-studies-in-a-potential}
formally introduces the event study design and discusses our definition
in relation to the applied literature. Section \ref{sec:Estimators-from-linear}
derives the estimands of two-way fixed effects regression, and introduces
sufficient assumptions for them to be causally interpretable. Section
\ref{sec:Alternative-method} develops our alternative estimator.
Section \ref{sec:Applications} illustrates our results using an empirical
example and Section \ref{sec:Conclusions} concludes. All proofs are
contained in the Online Appendix.

\subsection{Related Literature}

This paper makes two main contributions within an active literature
on the causal interpretations of two-way fixed effects models in settings
with staggered treatment adoption (\citealp{athey_imbens_staggered_adoption};
\citealp{kirill}; \citealp{callaway_santanna_ssrn2018}; \citealp{double_fe};
\citealp{goodman-bacon_difference--differences_2018}). Our paper
is also related to the traditional literature analyzing non-separable
panel and treatment effects models e.g. \citealp{heckman_ichimura_smith_todd_ecma1998,heckman_ichimura_todd_restud1997,blundell_jeea2004,semiparametric_did,victor_panel}.

The first main contribution of our paper is to interpret estimates
from two-way fixed effects specifications when researchers include
\textquotedblleft dynamic\textquotedblright{} indicators for time
relative to treatment and when treatment effects are heterogeneous
across adoption cohorts. We derive our results for a general class
of two-way fixed effects specifications where \textquotedblleft dynamic\textquotedblright{}
indicators can be flexibly specified as single relative periods $\ell$
or sets of relative periods $g$ (thus also capturing any \textquotedblleft static\textquotedblright{}
specification where all post-treatment indicators are collected in
a single set). This class of specifications encompasses all specifications
addressed by \citet{athey_imbens_staggered_adoption}, \citet{kirill},
\citet{callaway_santanna_ssrn2018}, \citet{double_fe} and \citet{goodman-bacon_difference--differences_2018}.

As a building block for the causal interpretation of estimates, we
define $CATT_{e,\ell}$, the cohort average treatment effects on the
treated as the cohort-specific average difference in outcomes relative
to never being treated. Our choice of a \textquotedblleft building
block\textquotedblright{} is governed by the counterfactual and the
type of heterogeneity of interest. This object coincides with the
\textquotedblleft group-time average treatment effect\textquotedblright{}
studied by \citet{callaway_santanna_ssrn2018} and is more granular
than the building block used by \citet{goodman-bacon_difference--differences_2018}
that is an average of $CATT_{e,\ell}$ over some relative period range.
\citet{athey_imbens_staggered_adoption} consider an alternate counterfactual
to never being treated: being treated at a different time. \citet{kirill}
implicitly assume away heterogeneity across cohorts within a relative
period, so their building block reduces to $ATT_{\ell}$. \citet{double_fe}
allow for heterogeneous treatment paths within a cohort over time
across \textquotedblleft groups\textquotedblright , thus their building
block is at the group level. We defer a discussion of assumptions
underlying the causal interpretation of these building blocks to Section
\ref{sec:Event-studies-in-a-potential}.

The second main contribution of our paper is to propose a simple regression-based
alternative estimation strategy that produces a more sensible estimand
than conventional two-way fixed effects models under heterogeneous
treatment effects. Our procedure is most similar to \citet{callaway_santanna_ssrn2018},
but has the following differences. First, in the setting where there
is no never-treated group, our method uses the last cohort to be treated
as a control group, whereas \citet{callaway_santanna_ssrn2018} use
the set of not-yet-treated cohorts. Our method and theirs thus rely
on different, but non-nested parallel trends assumptions. Second,
our estimation method can be cast as a regression specification and
thus may be more familiar to applied researchers. However, a third
difference is that the procedure of \citet{callaway_santanna_ssrn2018}
allows for conditioning on time-varying covariates. \citet{double_fe}
and \citet{goodman-bacon_difference--differences_2018} respectively
propose alternative estimators and diagnostic tools for estimation
of causal effects in staggered settings, but do not consider the estimation
of the dynamic path of treatment effects as we do. 

\section{Event studies design\label{sec:Event-studies-in-a-potential}}

In this section we first formalize the ``event studies design''.
As discussed in Section \ref{subsec:Relevance-in-the}, based on how
this term is deployed in the empirical literature, an \emph{event
study }design is a \emph{staggered adoption }design where units are
treated at different times, and there may or may not be never treated
units. It also nests a \emph{difference-in-differences }design, where
units are either first treated at time $t_{0}$ or never treated.

Specifically, we consider a setting with a random sample of $N$ units
observed over $T+1$ time periods, where $T$ is fixed. For each $i\in\{0,\dots,N\}$
and $t\in\{0,...,T\}$, we observe the outcome $Y_{i,t}$ and treatment
status $D_{i,t}\in\left\{ 0,1\right\} $: $D_{i,t}=1$ if $i$ is
treated in period $t$ and $D_{i,t}=0$ if $i$ is not treated in
period $t$. Throughout we assume that the observations $\{Y_{i,t},D_{i,t}\}_{t=0}^{T}$
are independent and identically distributed (i.i.d.). 

In the general case of event studies we focus on an \emph{absorbing
}treatment such that the treatment status over time is a non-decreasing
sequence of zeros and then ones, i.e. $D_{i,s}\leq D_{i,t}$ for $s<t$.
We can thus uniquely characterize a treatment path by the time period
of the initial treatment, denoted with $E_{i}=\min\left\{ t:D_{i,t}=1\right\} $.
If unit $i$ is never treated i.e. $D_{i,t}=0$ for all $t$, we set
$E_{i}=\infty$. Based on when they first receive the treatment, we
can also uniquely categorize units into disjoint cohorts $e$ for
$e\in\{0,\dots,T,\infty\}$, where units in cohort $e$ are first
treated at the same time $\{i:E_{i}=e\}$. 

We define $Y_{i,t}^{e}$ to be the potential outcome in period $t$
when unit $i$ is first treated in time period $e$. We define $Y_{i,t}^{\infty}$
to be the potential outcome if unit $i$ never receives the treatment,
which we call the ``baseline outcome''. Since the timing of the
initial treatment uniquely characterizes one's treatment path, we
can represent the observed outcome for unit $i$ as
\begin{align}
Y_{i,t}=Y_{i,t}^{E_{i}} & =Y_{i,t}^{\infty}+\sum_{0\leq e\leq T}(Y_{i,t}^{e}-Y_{i,t}^{\infty})\cdot\mathbf{1}\left\{ E_{i}=e\right\} .
\end{align}

For any treatment that is not absorbing, if we replace the treatment
status $D_{i,t}$ with an indicator for ever having received the treatment,
the new treatment is absorbing by construction.  Oftentimes the effect
of having ever received the treatment is of interest, as it captures
the path of treatment effects even though the treatment itself may
be transient. For example, \citet{deryugina_aej2017} is interested
in the fiscal cost for a county that has been hit by a hurricane.
While a hurricane itself may be transient, the impact of having had
a hurricane may not be transient, hence why \citet{deryugina_aej2017}
codes the year of the first hurricane experienced in a county as $E_{i}$.

In the next section, we use the notation developed above to define
the treatment effect of an event study design. 

\subsection{Defining treatment effect of an event study design}

In an event study design, we define the unit-level \emph{treatment
effect}  as the difference between the observed outcome relative to
the never-treated counterfactual outcome: $Y_{i,t}-Y_{i,t}^{\infty}$.
Recall that $Y_{i,t}^{\infty}$ denotes the potential outcome if unit
$i$ never receives the treatment. This particular counterfactual
outcome $Y_{i,t}^{\infty}$ is a reasonable ``baseline outcome'',
though other counterfactual outcomes may be of interest as well. For
example, \citet{athey_imbens_staggered_adoption} also consider the
treatment effect relative to the always-treated counterfactual outcome:
$Y_{i,t}-Y_{i,t}^{0}$. \citet{sianesi_restat2004} defines the unit-level
treatment effect to be relative to the not-yet-treated counterfactual
outcome: $Y_{i,t}-Y_{i,t}^{e}$ for $e>t$.

When dynamic treatment effects are of interest, empirical researchers
commonly report the coefficient estimate $\widehat{\mu}_{\ell}$ associated
with indicators for being $\ell$ periods relative to the treatment
in regression (\ref{eq:TWFE intro}) as an estimate for the average
lagged effect. To assess the causal interpretation of $\mu_{\ell}$,
we need ``building blocks'' for its decomposition, which are the
average of unit-level treatment effects at a given relative period
across units first treated at time $E_{i}=e$, i.e. units in the same
cohort $e$. We call this average the cohort-specific average treatment
effects on the treated, formally defined below. Later in Section \ref{sec:Estimators-from-linear}
we use them as building blocks for the interpretation of the relative
period coefficients $\mu_{\ell}$ from two-way fixed effects regressions.
\begin{defn}
\emph{\label{def:CATTel} }The cohort-specific average treatment effect
on the treated (CATT) $\ell$ periods from initial treatment is 
\begin{equation}
CATT_{e,\ell}=E[Y_{i,e+\ell}-Y_{i,e+\ell}^{\infty}\mid E_{i}=e].
\end{equation}
\end{defn}
Each $CATT_{e,\ell}$ represents the average treatment effect $\ell$
periods from the initial treatment for the cohort of units first treated
at time $e$. We shift from calendar time index $t$ to relative period
index $\ell$ which denotes the periods since treatment; for cohort
$e$, $\ell$ ranges from $-e$ to $T-e$ because we observe at most
$e$ periods before the initial treatment and $T-e$ periods after
the initial treatment. Relative periods allow us to compare cohorts
while holding their exposure to the treatment constant.

\subsection{Identifying assumptions}

With the above definitions, we formalize three potential identifying
assumptions for outcomes of interest in our event study design. The
first assumption is a generalized form of a parallel trends assumption.
The second assumption requires no anticipation of the treatment. The
third assumption imposes no variation across cohorts. For each assumption,
we first discuss its meaning and then compare it with similar assumptions
made in the literature interpreting two-way fixed effects regressions.
Later in Section \ref{sec:Estimators-from-linear} we interpret the
relative period coefficients $\mu_{\ell}$ from two-way fixed effects
regressions under different combinations of these assumptions.
\begin{assumption}
\emph{\label{assum:PT-baseline}(Parallel trends in baseline outcomes.)}
For all $s\neq t$, the $E[Y_{i,t}^{\infty}-Y_{i,s}^{\infty}\vert E_{i}=e]$
is the same for all $e\in supp(E_{i})$.
\end{assumption}
If an application includes never-treated units so that $\infty\in supp(E_{i})$,
we need to especially consider whether these never-treated units satisfy
the parallel trends assumption. Never-treated units are likely to
differ from ever-treated units in many ways, and may not share the
same evolution of baseline outcomes. If the never-treated units are
unlikely to satisfy the parallel trends assumption, then we should
exclude them from the estimation to avoid violation of this assumption.

While common in the applied literature, the parallel trends assumption
is strong and oftentimes violated. For example, \citet{ashenfelter_1978}
documented that participants in job training programs experience a
decline in earnings prior to the training period (Ashenfelter\textquoteright s
dip). The timing of job training is dependent on the evolution of
individual's baseline earnings, and this scenario therefore does not
satisfy the parallel trends assumption. Proposition \ref{prop: dyno FE coeff weights-1}
is the only result in this paper derived without this assumption,
but there is an active literature studying inference under violations
of the parallel trends assumption e.g. \citet{roth_rambachan_2020}. 

Our parallel trends assumption coincides with that of \citet{double_fe}.
One could substitute this assumption with a different identifying
assumption that baseline outcomes are mean independent of $E_{i}$
i.e. at each $t$, $E[Y_{i,t}^{\infty}\vert E_{i}=e]$ is the same
for all $e\in supp\left(E_{i}\right)$ and in particular is equal
to $E[Y_{i,t}^{\infty}]$. This stronger assumption is plausible when
the timing of treatment is indeed randomized, which is the assumption
used by \citet{athey_imbens_staggered_adoption}. By taking the ``fully
dynamic'' specification as their DGP, \citet{kirill} implicitly
assume this version of a parallel trends assumption. \citet{callaway_santanna_ssrn2018}
propose a weaker version that is conditional on covariates.  Finally,
for a particular estimand, \citet{goodman-bacon_difference--differences_2018}
identifies a weaker version that only requires a weighted average
of $E[Y_{i,t}^{\infty}-Y_{i,s}^{\infty}\vert E_{i}=e]$ (averaged
across cohorts) to be zero.
\begin{assumption}
\label{assum:no-anti}\emph{(No anticipatory behavior prior to treatment.)}
There is no treatment effect in pre-treatment periods i.e. $E[Y_{i,e+\ell}^{e}-Y_{i,e+\ell}^{\infty}\mid E_{i}=e]=0$
for all $e\in supp(E_{i})$ and all $\ell<0$.
\end{assumption}
Assumption \ref{assum:no-anti} requires potential outcomes in any\emph{
}$\ell$\emph{ }periods  before treatment to be equal to the baseline
outcome on average as in \citet{malani_reif_anti_jpube2015} and \citet{botosaru_gutierrez_jae2018}.
This is most plausible if the full treatment paths are not known to
units. If they have private knowledge of the future treatment path
they may change their behavior in anticipation and thus the potential
outcome prior to treatment may not represent baseline outcomes. For
example, \citet{hendren_consumption} shows that knowledge of future
job loss leads to decreases in consumption. If the periods with anticipation
behavior are known, then we may consider an alternative version of
Assumption \ref{assum:no-anti}, which holds for pre-periods in a
subset of pre-treatment periods. Depending on the application, it
may still be plausible to assume no anticipation until $K$ periods
before the treatment.

The no anticipation assumption proposed by \citet{athey_imbens_staggered_adoption}
is a deterministic condition which stipulates that $Y_{i,e+\ell}^{e}=Y_{i,e+\ell}^{\infty}$
for all units $i$ and $e$ and $\ell<0$. By taking the ``fully
dynamic'' specification as their DGP, \citet{kirill} allow anticipation
by including pre-trends indicators in the DGP. \citet{callaway_santanna_ssrn2018}
and \citet{goodman-bacon_difference--differences_2018} implicitly
assume no anticipation by using observed outcomes in time periods
before the initial treatment as the untreated potential outcomes.
\begin{assumption}
\emph{\label{assum:homogeneous-TE}(Treatment effect homogeneity.)
}For each relative period $\ell$, $CATT_{e,\ell}$ does not depend
on cohort $e$ and is equal to $ATT_{\ell}$. 
\end{assumption}
Assumption \ref{assum:homogeneous-TE} requires that each cohort experiences
the same path of treatment effects. Treatment effects need to be the
same across cohorts in every relative period for homogeneity to hold,
whereas for heterogeneity to occur, treatment effects only need to
differ across cohorts in one relative period. The assumption of treatment
effect homogeneity is therefore strong, and in Section \ref{subsec:Sources-of-treatment},
we describe how it can be violated in applied settings.

Our notion of treatment effect homogeneity does not preclude dynamic
treatment effects; it only imposes that cohorts share the same path
of treatment effects. The related literature sometimes formulates
restrictions on the dynamics of treatment effects as another notion
of treatment effect homogeneity. \citet{athey_imbens_staggered_adoption}
propose an assumption that ``restricts the heterogeneity of the treatment
effects over time,'' which implies $CATT_{e,\ell}$ can vary over
$e$ but not over $\ell$. \citet{kirill} refer to one type of treatment
effects heterogeneity as ``only across the time horizon,'' which
implies  $CATT_{e,\ell}$ can vary over $\ell$ but not over $e$.
\citet{callaway_santanna_ssrn2018} allow for ``arbitrary treatment
effect heterogeneity'' when $CATT_{e,\ell}$ varies across cohorts
and over time. Similarly, \citet{double_fe} describe treatment effects
that may be ``heterogeneous across groups and over time periods.''
\citet{goodman-bacon_difference--differences_2018} allow heterogenous
effects to either ``vary across units but not over time'' or ``vary
over time but not across units.'' The literature has not converged
on a single notion of treatment effects heterogeneity with time-varying
treatment. Since researchers are interested in dynamic treatment effects
when using a ``dynamic'' specification, we do not restrict the path
of treatment effects but rather use ``heterogeneity'' to describe
variation across cohorts only. 

\subsection{Relevance in the applied literature\label{subsec:Relevance-in-the} }

To gauge the empirical relevance of our results, we survey the estimation
methods used by the twelve papers collected by \citet{roth_pretest_2019}
from three leading economics journals that contain the phrase ``event
study'' in their main text.\footnote{We follow the selection criteria in \citet{roth_pretest_2019}: the
original sample consists of 70 total papers, but is further constrained
to these twelve papers with publicly available data and code. The
data and code are used to determine exactly the specification estimated
in these papers.} From this sample of applied papers, we learn what specifications
empirical researchers are actually using when estimating two-way fixed
effects regressions. Four papers in this sample consider the simple
setting where units either receive their first treatment at the same
time or never receive the treatment. The other eight papers in this
sample consider the more complex setting where treated units receive
their treatment at various times, and there may or may not be never
treated units. This observation suggests an \emph{event study} in
the applied literature nests two popular research designs: \emph{difference-in-differences
}design, but also the design where units receive their first treatments
at various times, which is our focus.\footnote{This setup is the same as the \emph{staggered adoption }design proposed
by \citet{athey_imbens_staggered_adoption}, but we keep the term
\emph{event study} because it is common in the applied literature. } 

We summarize the main specifications in this sample of twelve applied
papers in Table \ref{tab:Selected-publications}.\footnote{We focus on the first specification underlying the event study estimates
in each paper, which we view as a reasonable proxy for the main specification
in the paper.} The columns of Table \ref{tab:Selected-publications} collect key
properties of these specifications. In Section \ref{subsec:Possible-specifications},
we introduce a general class of specifications that encompasses all
of these specification. The estimates for relative period coefficients
$\mu_{\ell}$ from all these papers therefore fall under our analysis
in the next section. 

These papers demonstrated that event studies are used to address a
broad range of research questions. As an example of this literature,
\citet{bailey_goodman_bacon_aer2015} use the rollout of the first
Community Health Centers (CHCs) to study the longer-term health effects
of increasing access to primary care. As another example, \citet{tewari_aej2014}
uses variation in the timing of deregulation across states to estimate
the impact of financial development on homeownership. 

\section{Estimators from linear two-way fixed effects regression\label{sec:Estimators-from-linear}}

We consider a two-way fixed effects (FE) regression of the following
form, estimated on a panel of $i=1,\dots,N$ units for $t=0,1,\dots,T$
\emph{calendar }time periods:
\begin{equation}
Y_{i,t}=\alpha_{i}+\lambda_{t}+\sum_{g\in\mathcal{G}}\mu_{g}\mathbf{1}\{t-E_{i}\in g\}+\upsilon_{i,t}\label{eq:dynamic}
\end{equation}
Here $Y_{i,t}$ is the outcome of interest for unit $i$ at time $t$,
$E_{i}$ is the time for unit $i$ to initially receive a binary absorbing
treatment, and $\alpha_{i}$ and $\lambda_{t}$ are the unit and time
fixed effects. The set $\mathcal{G}$ collects disjoint sets $g$
of \emph{relative periods} $\ell\in[-T,T]$. We allow some relative
periods to be excluded from the specification and denote the excluded
set with $g^{excl}=\{\ell:\ell\not\in\underset{g\in\mathcal{G}}{\bigcup}g\}$.
We denote by $\mu_{g}$ the relative period coefficients from regression
(\ref{eq:dynamic}), i.e. the population regression coefficients.
Their corresponding OLS estimators are denoted by $\widehat{\mu}_{g}$
respectively. 

We are interested in the properties of $\mu_{g}$ when there are variations
in the initial treatment timing, and there may or may not be never-treated
units. Below in Section \ref{subsec:Possible-specifications} we illustrate
how the choice of $\mathcal{G}$ coincides with a large number of
specifications encountered in practice such as the ``fully dynamic''
specification. We next decompose $\mu_{g}$ in terms of $CATT_{e,\ell}$
when various combinations of the three identifying assumptions fail.
For Propositions \ref{prop: dyno FE coeff weights-1}-\ref{prop:dyno FE coeff weights no anti},
we state the results in terms of the general specification (\ref{eq:dynamic}).
To specialize these results to the ``fully dynamic'' specification
(\ref{eq:dynamic FE}), we note the corresponding decomposition would
replace bins with $g=\{\ell\}$ for each relative period $\ell$ included
in the specification and $g^{excl}$ would contain all excluded relative
periods. The decomposition remains unchanged though the summation
over $\ell\in g$ simplifies since each $g$ is a singleton. For Proposition
\ref{prop:dyno FE coeff weights homo}, the decomposition further
simplifies for the ``fully dynamic'' specification as discussed
below.

Researchers may assume that $\mu_{g}$ can be interpreted as a convex
average of $CATT_{e,\ell}$ for periods $\ell\in g$ from its corresponding
set $g$; they may further assume the underlying weights have policy-relevant
interpretation, e.g. weights depending on proportions of cohorts.
For example, \citet{bailey_goodman_bacon_aer2015} interpret them
as ``intention-to-treat effects'' of the treatment in a given relative
year. However, our results show that $\mu_{g}$ may not represent
the parameter of interest without strong assumptions such as treatment
effect homogeneity. Section \ref{subsec:Intuition-for-contamination}
provides intuition for these negative results, and demonstrate the
weights are actually non-linear functions of proportions of cohorts.
Section \ref{subsec:Invalidity-of-pretrend} illustrates how treatment
effects heterogeneity invalidates the pretrends test for a simple
three-period setting.

\subsection{Common specifications\label{subsec:Possible-specifications} }

Common specifications can be broadly categorized as either ``static''
or ``dynamic''. Static specifications estimate a single treatment
effect that is time invariant. In contrast, dynamic specifications
allow for non-parametric changes in the treatment effects over time.
Within dynamic specifications, researchers also need to address issues
of multi-collinearity, and may bin or trim distant relative periods.
All of these choices can be written as instances of (\ref{eq:dynamic})
with the correct specification of $\mathcal{G}$, meaning that our
results are applicable for a wide range of specifications employed
in the empirical literature.

To clarify how to specify $\mathcal{G}$ in regression (\ref{eq:dynamic}),
we define $D_{i,t}^{\ell}\coloneqq\mathbf{1}\{t-E_{i}=\ell\}$ to
be an indicator for unit $i$ being $\ell$ periods away from initial
treatment at calendar time $t$. For never-treated units $E_{i}=\infty$,
we set $D_{i,t}^{\ell}=0$ for all $\ell$ and all $t$. We can represent
the relative period bin indicator as
\begin{equation}
\mathbf{1}\{t-E_{i}\in g\}=\sum_{\ell\in g}\mathbf{1}\{t-E_{i}=\ell\}=\sum_{\ell\in g}D_{i,t}^{\ell}.
\end{equation}

\textbf{Static specification. }For a ``static'' specification $\mathcal{G}$
contains a single element equal to $g=[0,T]$. The indicator $\mathbf{1}\{t-E_{i}\in g\}$
is equivalent to an indicator for whether unit $i$ has received its
initial treatment by $t$: $\mathbf{1}\{E_{i}\leq t\}$. The ``static''
specification thus takes the following form
\begin{equation}
Y_{i,t}=\alpha_{i}+\lambda_{t}+\mu_{g}\sum_{\ell\geq0}D_{i,t}^{\ell}+\upsilon_{i,t}
\end{equation}
and the corresponding set of excluded relative periods is $g^{excl}=[-T,-1]$.

\textbf{Dynamic specification. }``Dynamic'' specifications encompass
any specifications of $\mathcal{G}$ where $\mathcal{G}$ contains
more than one element, thus treatment effects are allowed to vary
over time non-parametrically. In its most flexible form, a ``fully
dynamic'' specification takes the following form 
\begin{equation}
Y_{i,t}=\alpha_{i}+\lambda_{t}+\sum_{\ell=-K}^{-2}\mu_{\ell}D_{i,t}^{\ell}+\sum_{l=0}^{L}\mu_{\ell}D_{i,t}^{\ell}+\upsilon_{i,t}\label{eq:dynamic FE}
\end{equation}
and the corresponding set of excluded relative periods is $g^{excl}=\{-T,\dots,-K-1,-1,L+1,\dots,T\}$. 

Excluding some relative periods from the ``fully dynamic'' specification
is necessary to avoid multi-collinearity, either among the relative
period indicators $D_{i,t}^{\ell}$, or with the unit and time fixed
effects. For example, when there are no never-treated units i.e. $\infty\not\in supp(E_{i})$
but with a panel balanced in calendar time, we need to exclude at
least two relative period indicators in $\mathcal{G}$. These collinearities
are discussed by \citet{kirill}: one multi-collinearity comes from
the relative period indicators summing to one for every unit $\sum_{\ell\in[-T,T]}D_{i,t}^{\ell}=1$,
and the other multi-collinearity comes from the linear relationship
between two-way fixed effects and the relative period indicators,
namely $t-E_{i}=\ell$. 

Excluding relative periods close to the initial treatment is common
in practice. Normalizing relative to the period prior to treatment
is the most common - six out of the eight papers we survey do so,
as reflected in the above specification where we drop $D_{i,t}^{-1}$.
The remaining two papers exclude $D_{i,t}^{0}$. 

Excluding distant relative periods is however less common (only one
of the eight papers we survey does so). Instead researchers ``bin''
or ``trim'' distant relative periods. For ``binning'', researchers
bin distant relative periods into $[-T,-K)$ and $(L,T]$ and estimate
a ``binned'' specification 
\begin{equation}
Y_{i,t}=\alpha_{i}+\lambda_{t}+\beta\cdot\sum_{\ell<-K}D_{i,t}^{\ell}+\sum_{\ell=-K}^{-2}\mu_{\ell}D_{i,t}^{\ell}+\sum_{l=0}^{L}\mu_{\ell}D_{i,t}^{\ell}+\gamma\cdot\sum_{\ell>L}D_{i,t}^{\ell}+\upsilon_{i,t}\label{eq:dynamic FE bin}
\end{equation}
without excluding any distant relative periods so that $g^{excl}=\{-1\}$.
For ``trimming'', researchers trim their panel to be balanced in
relative periods. 

Neither \textquotedblleft binning\textquotedblright{} nor \textquotedblleft trimming\textquotedblright{}
resolves the issues of contamination discussed below (i.e. the possibility
that treatment effects from other periods affect the estimate for
a given $\mu_{g}$). We show the contamination issue for the general
specification (\ref{eq:dynamic}) encompasses both practices. For
a given coefficient in the dynamic specification, \textquotedblleft trimming\textquotedblright{}
does however mechanically remove any treatment effects from the relative
periods \textquotedblleft trimmed\textquotedblright{} from the specification.
For the static specification put forth in \citet{kirill}, they noted
that \textquotedblleft trimming\textquotedblright{} also does not
resolve the contamination issue they identified with the static specification.

\subsection{Interpreting the coefficients under no assumptions}

First we show that without any assumptions, we can write $\mu_{g}$
as a linear combination of differences in trends.
\begin{prop}
\label{prop: dyno FE coeff weights-1} The population regression coefficient
on relative period bin $g$ is a linear combination of differences
in trends from its own relative period $\ell\in g$, from relative
periods $\ell\in g'$ belonging to other bins $g'\neq g$ but included
in the specification, and from relative periods excluded from the
specification $\ell\in g^{excl}$:
\begin{align}
\mu_{g}= & \sum_{\ell\in g}\sum_{e}\omega_{e,\ell}^{g}\left(E[Y_{i,e+\ell}-Y_{i,0}^{\infty}\vert E_{i}=e]-E[Y_{i,e+\ell}^{\infty}-Y_{i,0}^{\infty}]\right)\label{eq:own rel time}\\
 & +\sum_{g'\neq g,g'\in\mathcal{G}}\sum_{\ell\in g'}\sum_{e}\omega_{e,\ell}^{g}\left(E[Y_{i,e+\ell}-Y_{i,0}^{\infty}\vert E_{i}=e]-E[Y_{i,e+\ell}^{\infty}-Y_{i,0}^{\infty}]\right)\label{eq:other incl rel time}\\
 & +\sum_{\ell\in g^{excl}}\sum_{e}\omega_{e,\ell}^{g}\left(E[Y_{i,e+\ell}-Y_{i,0}^{\infty}\vert E_{i}=e]-E[Y_{i,e+\ell}^{\infty}-Y_{i,0}^{\infty}]\right).\label{eq:other excl rel time}
\end{align}
We use the superscript $g$ to associate the weight $\omega_{e,\ell}^{g}$
with the coefficient $\mu_{g}$. The weight $\omega_{e,\ell}^{g}$
is equal to the population regression coefficient on $\mathbf{1}\{t-E_{i}\in g\}$
from regressing $D_{i,t}^{\ell}\cdot\mathbf{1}\left\{ E_{i}=e\right\} $
on all bin indicators $\{\mathbf{1}\{t-E_{i}\in g\}\}_{g\in\mathcal{G}}$
included in the specification (\ref{eq:dynamic}) and two-way fixed
effects. 
\end{prop}
The above proposition is a direct result of regression mechanics.
We provide an intuitive derivation for the closed-form expressions
for the weights using the classical ``omitted variables bias formula''
in Section \ref{subsec:Intuition-for-contamination}. We defer the
formal derivation to the Appendix. Here we mention the following four
properties of the weights $\omega_{e,\ell}^{g}$. 
\begin{itemize}
\item For relative periods of $\mu_{g}$'s own bin i.e. $\ell\in g$, their
associated weights as displayed in (\ref{eq:own rel time}) sum to
 one $\sum_{\ell\in g}\sum_{e}\omega_{e,\ell}^{g}=1$. 
\item For relative periods belonging to some other bin included in (\ref{eq:dynamic})
i.e. $\ell\in g'$ for $g'\neq g$ and $g'\in\mathcal{G}$, their
associated weights as displayed in (\ref{eq:other incl rel time})
sum to  zero $\sum_{\ell\in g'}\sum_{e}\omega_{e,\ell}^{g}=0$ for
each bin $g'$. 
\item For relative periods not included in $\mathcal{G}$, their associated
weights as displayed in (\ref{eq:other excl rel time}) sum to  negative
one $\sum_{\ell\in g^{excl}}\sum_{e}\omega_{e,\ell}^{g}=-1$. 
\item If there are never-treated units i.e. $\infty\in supp(E_{i})$, we
have $\omega_{\infty,\ell}^{g}=0$ for all $g$ and $\ell$. 
\end{itemize}
We can easily estimate the weights $\omega_{e,\ell}^{g}$ for any
given specification of $\mathcal{G}$ using the following auxiliary
regression: 
\begin{equation}
D_{i,t}^{\ell}\cdot\mathbf{1}\left\{ E_{i}=e\right\} =\alpha_{i}+\lambda_{t}+\sum_{g\in\mathcal{G}}\omega_{e,\ell}^{g}\mathbf{1}\{t-E_{i}\in g\}+\upsilon_{i,t}\label{eq:dynamic-weights}
\end{equation}
which regresses $D_{i,t}^{\ell}\cdot\mathbf{1}\left\{ E_{i}=e\right\} $
on all bin indicators included in regression (\ref{eq:dynamic}) and
two-way fixed effects. 

All of the above properties can be extended to a case where covariates
are added to regression (\ref{eq:dynamic}) by partialling out the
covariates before proceeding. In other words, the terms in parentheses
in (\ref{eq:own rel time}), (\ref{eq:other incl rel time}) and (\ref{eq:other excl rel time})
would be replaced by terms for which the covariates are partialled
out. The weights can be estimated by controlling for covariates in
regression (\ref{eq:dynamic-weights}) the same way as they are controlled
for in the original regression. 

However, covariates complicate the interpretations of $\mu_{g}$ in
terms of $CATT_{e,\ell}$ as we describe below in Proposition 2-4.
Depending on how covariates are controlled for in regression (\ref{eq:dynamic}),
we may need an additional assumption that the counterfactual trends
are linear in the time-varying covariates $X_{i,t}$. We leave a full
investigation of the introduction of covariates for future work.

\subsection{Interpreting the coefficients under parallel trends assumption only}
\begin{prop}
\label{prop: dyno FE coeff weights} Under Assumption \ref{assum:PT-baseline}
(parallel trends) only, the population regression coefficient on the
indicator for relative period bin $g$ is a linear combination of
$CATT_{e,\ell\in g}$ as well as $CATT_{e,\ell'}$ from other relative
periods $\ell'\not\in g$, with the same weights stated in Proposition
\ref{prop: dyno FE coeff weights-1}:
\begin{equation}
\mu_{g}=\sum_{\ell\in g}\sum_{e}\omega_{e,\ell}^{g}CATT_{e,\ell}+\sum_{g'\neq g,g'\in\mathcal{G}}\sum_{\ell'\in g'}\sum_{e}\omega_{e,\ell'}^{g}CATT_{e,\ell'}+\sum_{\ell'\in g^{excl}}\sum_{e}\omega_{e,\ell'}^{g}CATT_{e,\ell'}.\label{eq:dyno-expression-w-anti}
\end{equation}
\end{prop}
Under Assumption \ref{assum:PT-baseline}, the terms in Proposition
\ref{prop: dyno FE coeff weights-1} reduce to a linear combination
of the causally interpretable building blocks $CATT_{e,\ell}$ as
follows: 
\begin{equation}
E[Y_{i,e+\ell}-Y_{i,0}^{\infty}\vert E_{i}]-E[Y_{i,e+\ell}^{\infty}-Y_{i,0}^{\infty}]=CATT_{e,\ell}+\underbrace{E[Y_{i,t}^{\infty}-Y_{i,0}^{\infty}\vert E_{i}]-E[Y_{i,t}^{\infty}-Y_{i,0}^{\infty}]}_{=0}
\end{equation}
for $t=e+\ell$. However, two issues for interpretability remain.
First, the coefficient $\mu_{g}$ can be written as an average of
not only $CATT_{e,\ell}$ from own periods $\ell\in g$, but also
$CATT_{e,\ell^{\prime}}$ from other periods. Second, the weights
are still non-linear functions of the distribution of cohorts, same
as those in (\ref{eq:dynamic-weights}), and they are not restricted
to lie in $[0,1]$. 

The properties of these weights as described following Proposition
\ref{prop: dyno FE coeff weights-1} imply that contamination from
other periods wanes once we impose restrictions on treatment effects.
In the next two subsections we illustrate how that can happen.

\subsection{Interpreting the coefficients under parallel trends and no anticipation
assumptions }
\begin{prop}
\label{prop:dyno FE coeff weights no anti}If Assumption \ref{assum:PT-baseline}
(parallel trends) holds and Assumption \ref{assum:no-anti} (no anticipatory
behavior in all periods before the initial treatment) holds, the population
regression coefficient $\mu_{g}$ is a linear combination of post-treatment
$CATT_{e,\ell'}$ for all $\ell'\geq0$, with the same weights stated
in Proposition \ref{prop: dyno FE coeff weights-1}:
\begin{equation}
\mu_{g}=\sum_{\ell'\in g,\ell'>0}\sum_{e}\omega_{e,\ell}^{g}CATT_{e,\ell}+\sum_{g'\neq g,g'\in\mathcal{G}}\sum_{\ell'\in g',\ell'>0}\sum_{e}\omega_{e,\ell'}^{g}CATT_{e,\ell'}+\sum_{\ell'\in g^{excl},\ell'>0}\sum_{e}\omega_{e,\ell'}^{g}CATT_{e,\ell'}.\label{eq:dyno-no-anti}
\end{equation}
\end{prop}
Once we restrict pre-treatment $CATT_{e,\ell\leq0}$ to be zero under
the no anticipatory behavior assumption, the expression for $\mu_{g}$
simplifies as terms involving $CATT_{e,\ell\leq0}$ drop out. However,
the second term in the expression for $\mu_{g}$ remains unless we
further impose treatment effect homogeneity for its summands to cancel
out each other. Thus, $\mu_{g}$ may be non-zero for pre-treatment
periods even if parallel trends holds.

This result immediately implies a shortcoming of using pre-treatment
coefficients (i.e. $\mu_{g}$ where $g$ contains only leads to the
treatment $\ell<0$) to test for pretrends. Under the no anticipatory
behavior assumption, cohort-specific treatment effects prior to treatment
are all zero: $CATT_{e,\ell}=0$ for all $\ell<0$. Therefore, any
linear combination of these $CATT_{e,\ell}$ is also zero. However,
$\mu_{g}$ is a function of post-treatment $CATT_{e,\ell'\geq0}$
as well, even when $g$ only contains elements with $\ell<0$. We
revisit this implication in greater depth in Section \ref{subsec:Invalidity-of-pretrend}.
\citet{callaway_santanna_ssrn2018} provides alternative tests for
pretrends that do not suffer from this drawback.

\subsubsection{Sources of treatment effect heterogeneity\label{subsec:Sources-of-treatment}}

Since treatment effects heterogeneity violates Assumption \ref{assum:homogeneous-TE}
and can alter how we interpret $\mu_{g}$, it is important to think
through when different cohorts likely experience different paths of
treatment effect. Such heterogeneity could arise for many reasons.
For example, cohorts may differ in their covariates, which affect
how they respond to treatment. We will explore a concrete example
in our application: if treatment effects differ with age, and there
is variation in age across units first treated at different times,
we will have heterogeneous effects (see Section \ref{sec:Applications}
for details). After controlling for covariates, cohorts may still
vary in their responses to the treatment if units select their initial
treatment timing based on treatment effects. This source of heterogeneity
is still compatible with our parallel trends assumption, which only
rules out selection in the initial treatment timing based on the evolution
of the baseline outcome. In addition to these two sources of heterogeneity,
treatment effects may vary across cohorts due to calendar time-varying
effects (e.g. macroeconomic conditions could govern the effects on
labor market outcomes across cohorts).

\subsection{Interpreting the coefficients under parallel trends and treatment
effect homogeneity}
\begin{prop}
\label{prop:dyno FE coeff weights homo}If Assumption \ref{assum:PT-baseline}
(parallel trends) holds and Assumption \ref{assum:homogeneous-TE}
(treatment effect homogeneity) holds, then $CATT_{e,\ell}=ATT_{\ell}$
is constant across $e$ for a given $\ell$, and the population regression
coefficient $\mu_{g}$ is equal to a linear combination of $ATT_{\ell\in g}$,
as well as $ATT_{\ell'\not\in g}$ from other relative periods:
\begin{equation}
\mu_{g}=\sum_{\ell\in g}\omega_{\ell}^{g}ATT_{\ell}+\sum_{g'\neq g}\sum_{\ell'\in g'}\omega_{\ell'}^{g}ATT_{\ell'}+\sum_{\ell'\in g^{excl}}\omega_{\ell'}^{g}ATT_{\ell'}\label{eq:dyno-lead-homo}
\end{equation}
The weight $\omega_{\ell}^{g}=\sum_{e}\omega_{e,\ell}^{g}$ sums over
the weights $\omega_{e,\ell}^{g}$ from Proposition \ref{prop: dyno FE coeff weights-1},
and is equal to the population regression coefficient from the following
auxiliary regression: 
\begin{equation}
D_{i,t}^{\ell}=\alpha_{i}+\lambda_{t}+\sum_{g\in\mathcal{G}}\omega_{e,\ell}^{g}\mathbf{1}\{t-E_{i}\in g\}+\upsilon_{i,t}\label{eq:dynamic-weights-1}
\end{equation}
which regresses $D_{i,t}^{\ell}$ on all bin indicators included in
regression (\ref{eq:dynamic}) and two-way fixed effects. 
\end{prop}
We note that even under treatment effect homogeneity $\mu_{\ell}$
can still be contaminated by treatment effects from the excluded periods.
This contamination, however, can be avoided by adjusting the specification
to only exclude periods with zero treatment effect.

For specifications with relative time bins, we note that there can
still be contamination from other bins as suggested by the second
term of expression (\ref{eq:dyno-lead-homo}). A sufficient condition
to avoid such contamination would be to group relative periods $\ell'$
into a bin only when their effects are the same since their weights
$\omega_{\ell'}^{g}$ would sum to zero.

For the ``fully dynamic'' specification (\ref{eq:dynamic FE}) where
all $g$'s are singletons of relative time periods, the weight $\omega_{\ell'}^{\ell}$
is zero for each relative period $\ell'\ne\ell$ that is included
in the specification. The decomposition therefore simplifies to
\begin{equation}
\mu_{\ell}=ATT_{\ell}+\sum_{\ell'\in g^{excl}}\omega_{\ell'}^{g}ATT_{\ell'}.\label{eq:dyno-FE-correct}
\end{equation}

\subsection{Intuition for contamination\label{subsec:Intuition-for-contamination}}

Proposition \ref{prop:dyno FE coeff weights no anti} demonstrates
that even under the assumptions of parallel trends and no anticipation,
estimates $\mu_{g}$ can still be contaminated by treatment effects
from other periods. In this section we explain the intuition behind
why\emph{ }this contamination occurs for the ``fully dynamic'' specification
(\ref{eq:dynamic FE}). A decomposition of $\mu_{\ell}$ into a weighted
average of $CATT_{e,\ell'}$ demonstrates that contamination is driven
by the interaction of two elements: the weights and $CATT_{e,\ell'}$.
The weights underlying the contamination are non-linear functions
of the distribution of the cohorts. We do not attempt to provide a
heuristic for determining the magnitude of the weights, but instead
describe how to estimate the weights and later on in Section \ref{sec:Alternative-method}
how to estimate each $CATT_{e,\ell}$. This allows researchers to
directly determine the degree of contamination in their application.
Our publicly-available Stata package {\footnotesize{}\texttt{eventstudyweights}}
automates the estimation of these weights using the panel dataset
underlying any given specification of (\ref{eq:TWFE intro}).

We apply the familiar omitted variable bias (OVB) formula to arrive
at our decomposition. In an event study where individuals receive
the treatment at different times, the panel can never be balanced
in both calendar time and time relative to the initial treatment.
As a result, the relative time indicators are still correlated even
after controlling for unit and time fixed effects in a two-way fixed
effects regression. We use the saturated regression and the OVB formula
to illustrate how this correlation leads to contamination. We defer
its formal derivation to Appendix \ref{sec:Proofs}. 

Under the parallel trends assumption only, the saturated regression
is

\begin{align}
Y_{i,t}= & \sum_{e}E[Y_{i,0}^{\infty}\mid E_{i}=e]\cdot\mathbf{1}\{E_{i}=e\}+\sum_{s}E[Y_{i,s}^{\infty}-Y_{i,0}^{\infty}]\cdot\mathbf{1}\{t=s\}\nonumber \\
 & +\sum_{\ell'\in g^{incl}}\sum_{e\in\mathcal{I}_{\ell'}}CATT_{e,\ell'}\cdot\left(D_{i,t}^{\ell'}\cdot\mathbf{1}\left\{ E_{i}=e\right\} \right)\nonumber \\
 & +\sum_{\ell'\in g^{excl}}\sum_{e\in\mathcal{I}_{\ell'}}CATT_{e,\ell'}\cdot\left(D_{i,t}^{\ell'}\cdot\mathbf{1}\left\{ E_{i}=e\right\} \right)+\epsilon_{i,t}\label{eq:dynamic saturated}
\end{align}
where the regressors are cohort fixed effects, time fixed effects,
and cohort-specific relative time indicators. Furthermore, let $g^{incl}$
collect the relative time included in (\ref{eq:dynamic FE}). The
coefficient associated with the cohort-specific relative time indicator
$D_{i,t}^{\ell}\cdot\mathbf{1}\left\{ E_{i}=e\right\} $ is the cohort-specific
average treatment effects $CATT_{e,\ell}$. To decompose the coefficient
$\mu_{\ell}$ from (\ref{eq:dynamic FE}) in terms of this saturated
regression (\ref{eq:dynamic saturated}), the OVB formula multiplies
the coefficients in the saturated regression (\ref{eq:dynamic saturated}),
$CATT_{e,\ell'}$, with the regression coefficients from (\ref{eq:dynamic-weights}),
$\omega_{e,\ell'}^{\ell}$, which leads to the following decomposition
for $\mu_{\ell}$ as a linear combination of $CATT_{e,\ell'}$:

\begin{equation}
\mu_{\ell}=\sum_{e,\ell'}\omega_{e,\ell'}^{\ell}CATT_{e,\ell'}.\label{eq:heterogenous decomp}
\end{equation}
 Since $\omega_{e,\ell'}^{\ell}$ is equal to a regression coefficient
from (\ref{eq:dynamic-weights}), we can write it as 
\begin{equation}
\omega_{e,\ell'}^{\ell}=(\sigma_{e,\cdot})^{\intercal}\mathbf{\Delta}_{e+\ell'}A_{\ell}^{-1}.\label{eq:heterogenous weight decomp}
\end{equation}
Below we briefly comment on each of the three elements in the above
expression to highlight how they depend on the distribution of the
cohorts. We defer their detailed definitions and derivations to Appendix
\ref{subsec:Expressions-weights}.
\begin{itemize}
\item $\sigma_{e,\cdot}$ is a vector of the covariance between cohort $e$
and the other cohorts, namely $Cov\left(\mathbf{1}\{E_{i}=e\},\mathbf{1}\{E_{i}=e'\}\right)$.
This term thus scales quadratically in the share of cohort $e$, and
is small for small cohorts.
\item $\mathbf{\Delta}_{t}$ is a matrix of demeaned relative time indicators.
The entry that corresponds to cohort $e'$ and relative time indicator
$D_{i,t}^{\ell}$ is $E[D_{i,t}^{\ell}\mid E_{i}=e']-\frac{1}{T+1}\mathbf{1}\left\{ e'\in\mathcal{I}_{\ell}\right\} $.
When $T$ is large, i.e. the panel is long, the second term is small
and this entry is therefore approximately equal to the relative time
indicator.
\item $A_{\ell}^{-1}$ is the row of $A^{-1}$ that corresponds to the relative
time indicator $D_{i,t}^{\ell}$ for $A$ the covariance matrix of
demeaned relative time indicators. Specifically, the entry in $A$
that corresponds to the covariance between demeaned $D_{i,t}^{\ell}$
and $D_{i,t}^{\ell'}$ is 
\begin{equation}
\sum_{t}Cov\left(D_{i,t}^{\ell},D_{i,t}^{\ell'}\right)-\frac{1}{T+1}Cov\left(\mathbf{1}\left\{ E_{i}\in\mathcal{I}_{\ell}\right\} ,\mathbf{1}\left\{ E_{i}\in\mathcal{I}_{\ell'}\right\} \right).
\end{equation}
Within any time period $D_{i,t}^{\ell}$ and $D_{i,t}^{\ell'}$ are
negatively correlated because no cohorts can be in these two relative
times at the same time. The second covariance term is in general also
non-zero because being in one cohort predicts (not) being in another
cohort. Therefore $A$ is in general not a diagonal matrix and $A^{-1}$
would depend on the distribution of the cohorts non-linearly.
\end{itemize}
The three elements of (\ref{eq:heterogenous weight decomp}) demonstrate
the weights are non-linear functions of the distribution of the cohorts,
and they are in general non-zero. Nonetheless these weights can be
estimated easily by the auxiliary regression (\ref{eq:dynamic-weights}). 

\subsection{Invalidity of pretrend tests based on pre-period coefficients. \label{subsec:Invalidity-of-pretrend}}

Contamination undermines the practice of testing for pretrends using
pre-period coefficients. Proposition \ref{prop:dyno FE coeff weights no anti}
implies that when effects are not homogenous across cohorts, it is
problematic to interpret non-zero estimates for $\mu_{g}$ as evidence
for pretrends, where the set $g$ contains some leads $\ell<0$. Proposition
\ref{prop:dyno FE coeff weights homo} implies that even with homogeneous
treatment effect, if the effects associated with the excluded periods
are not zero, then contamination may still occur. Therefore without
strong assumptions, pre-period coefficients should not be used to
test for pretrends because contamination can lead to estimates that
are non-zero in the absence of pretrends or zero in the presence of
pre-trends.

Testing for pretrends using pre-period coefficients is commonly used
in practice. As an example, \citet{he_wang_aej2017} mention \textquotedblleft the
estimated coefficients of the leads of treatments, i.e. $\delta_{k}$
for all $k\leq-2$ are statistically indifferent from zero\textquotedblright{}
as evidence for lack of pretrends. As another example, \citet{Chetty_qje2014}
assert ``there is no trend toward higher individual pension contributions
prior to year 0 ... as one would expect if individuals\textquoteright{}
tastes for saving were changing around the job switch'' based on
pre-period coefficient estimates. These tests are only appropriate
when the authors are willing to make strong assumptions.

To provide further intuition for why this test is not meaningful without
additional assumptions we walk through a simple example of the fully
dynamic specification. Consider a balanced panel with $T=2$ and cohorts
$E_{i}\in\{1,2\}$. There are at least two multi-colinearities from
including all four relative time indicators. To form the fully dynamic
specification we include $g^{incl}=\{-2,0\}$ and exclude $g^{excl}=\{-1,1\}$:

\begin{equation}
Y_{i,t}=\alpha_{i}+\lambda_{t}+\sum_{\ell\in\{-2,0\}}\mu_{\ell}D_{i,t}^{\ell}+\upsilon_{i,t}.\label{eq:weights ex rel times-1}
\end{equation}
The choice of $g^{excl}$ is based on the common practice of normalizing
relative to the $-1$ period and distant lags. 

When there are no never treated units, we can express the pre-trend
coefficient $\mu_{-2}$ in terms of $CATT$s:
\begin{align}
\mu_{-2}= & \underbrace{CATT_{2,-2}}_{\text{own period}}+\underbrace{\frac{1}{2}CATT_{1,0}-\frac{1}{2}CATT_{2,0}}_{\ell'\in g^{incl},\ell'\neq-2}+\underbrace{\frac{1}{2}CATT_{1,1}-CATT_{1,-1}-\frac{1}{2}CATT_{2,-1}}_{\ell'\in g^{excl}}
\end{align}
It is apparent the weights maintain the structure described in Proposition
\ref{prop: dyno FE coeff weights-1}. Without any anticipation effect,
the effects $CATT_{e,\ell<0}$ are zero and thus we expect $\mu_{-2}$
to be zero regardless of the cohort shares. With homogeneous treatment
effect, cohorts 1 and 2 experience the same treatment effect at relative
time 0 so that $CATT_{1,0}$ and $CATT_{2,0}$ cancel. But even with
homogeneous treatment effect, the last term reflects the role of excluded
periods as $CATT_{1,-1}$, $CATT_{2,-1}$ and $CATT_{1,1}$ receive
non-zero weights. If there is any lagged effect and $CATT_{1,1}$
is non-zero, the coefficient $\mu_{-2}$ would be non-zero even without
any anticipation effect. Note such behavior is independent of the
distribution of the two cohorts.

We can further introduce never treated units to our example to demonstrate
how the weight $\omega_{e,\ell'}^{-2}$ can be non-linear in the distribution
of cohorts while maintaining the structure described in Proposition
\ref{prop: dyno FE coeff weights-1}. In Figure \ref{fig:Weights-on-mu_m2}
we plot the weights $\omega_{e,\ell'}^{-2}$ as we vary the distribution
of cohorts. Specifically, we vary the total share of treated cohorts
(shown on the $x$-axis), holding the shares of cohort 1 and 2 equal
to each other and setting the remaining to be the share of never treated
units. For any distribution of cohorts, we have $\omega_{2,-2}^{-2}=1$
(not pictured). Panel (a) shows the weights for the included period
$\ell'=0$ while panel (b) shows the weights for the excluded periods
$\ell'=-1$ or 1.

The example shown in Figure \ref{fig:Weights-on-mu_m2} provides a
visualization of three takeaways regarding the contamination in the
pre-trend coefficient $\mu_{-2}$. First, both panels show that weights
are a non-linear function of cohort shares. Second, panel (a) confirms
the structure for weights associated with $\ell'\neq\ell$ but $\ell'\in g^{incl}$
as described in Proposition \ref{prop: dyno FE coeff weights-1},
namely $\sum_{e}\omega_{e,\ell'}^{-2}=0$ for $\ell'\neq-2$. However,
these weights have non-zero magnitude. When the effect is homogenous
across cohorts, the contaminations are equal to $\sum_{e}\omega_{e,\ell'}^{-2}ATT_{\ell'}$
and cancel each other out. In contrast, when effects are heterogeneous
the different $CATT_{e,\ell'}$ will not necessarily cancel and will
contaminate the estimate for $\mu_{-2}$. Third, panel (b) confirms
the structure for weights associated with excluded periods as described
in Proposition \ref{prop: dyno FE coeff weights-1}, namely $\sum_{e}\omega_{e,\ell'\in g^{excl}}^{-2}=-1$.
In other words, contaminations from excluded periods can be thought
of as a type of \textquotedblleft normalization\textquotedblright :
a weighted average of excluded $CATT_{e,\ell'}$ is subtracted off
the estimated treatment effect. However because these weights are
not contained in $[0,-1]$, this average may lie outside of the convex
hull of $CATT_{e,\ell'}$ for excluded periods. The latter issue can
be alleviated by an assumption that $CATT_{e,\ell'}$ are the same
in all excluded periods; we can fully avoid contamination from excluded
period treatment effects by assuming all associated $CATT_{e,\ell'}$
are equal to zero.

\section{Alternative estimation method \label{sec:Alternative-method}}

We propose a new estimation method that is robust to treatment effects
heterogeneity. The goal of our method is to estimate a weighted average
of $CATT_{e,\ell}$ for $\ell\in g$ with reasonable weights, namely
weights that sum to one and are non-negative. In particular, we focus
on the following weighted average of $CATT_{e,\ell}$, where the weights
are shares of cohorts that experience at least $\ell$ periods relative
to treatment, normalized by the size of $g$:
\begin{equation}
\nu_{g}=\frac{1}{\left|g\right|}\sum_{\ell\in g}\sum_{e}CATT_{e,\ell}Pr\{E_{i}=e\mid E_{i}\in[-\ell,T-\ell]\}.
\end{equation}
One can aggregate $CATT_{e,\ell}$ to form other parameters of interest,
such as those proposed by \citet{callaway_santanna_ssrn2018}. We
focus on the above aggregation $\nu_{g}$ since our goal is to improve
the non-convex and non-zero weighting in $\mu_{g}$. The weights in
$\nu_{g}$ are guaranteed to be convex and have an interpretation
as the representative shares corresponding to each $CATT_{e,\ell}$.
Thus, our alternative estimator $\widehat{\nu}_{g}$ improves upon
the two-way fixed effects estimator $\widehat{\mu}_{g}$ by estimating
an interpretable weighted average of $CATT_{e,\ell\in g}$. 

Our method proceeds by replacing each component in $\nu_{g}$ with
its consistent estimator. We first estimate each $CATT_{e,\ell}$
using an interacted two-way fixed effects regression, then estimate
the weight $Pr\{E_{i}=e\mid E_{i}\in[-\ell,T-\ell]\}$ using their
sample analogs. In the final step, we average over the cohort-specific
estimates associated with relative period $\ell$. This method has
a similar flavor as the method proposed by \citet{gibbons_serrato_urbancic_jem2018}.
They first use an interacted model to estimate the treatment effect
for each fixed effect group; the resulting group-specific estimates
are averaged to provide the ATE. Their method improves fixed effects
regressions in a cross-sectional setting, and our method builds on
theirs by improving two-way fixed effects regressions in a panel setting.
We therefore follow their terminology in calling our alternative estimator
an ``interaction-weighted'' estimator. 

\subsection{Interaction-weighted estimator\label{subsec:Interaction-weighted-estimator}}

We describe the estimation procedure in three steps (with more detailed
definitions stated in Definition \ref{def:The-IW-estimator} of Online
Appendix \ref{sec:Proofs}). 

\textbf{Step 1. }We estimate $CATT_{e,\ell}$ using a linear two-way
fixed effects specification that interacts relative period indicators
with cohort indicators, excluding indicators for cohorts from some
set $C$: 
\begin{align}
Y_{i,t} & =\alpha_{i}+\lambda_{t}+\sum_{e\not\in C}\sum_{\ell\neq-1}\delta_{e,\ell}(\mathbf{1}\{E_{i}=e\}\cdot D_{i,t}^{\ell})+\epsilon_{i,t}.\label{eq: saturated event studies model}
\end{align}
The exact specification depends on the cohort shares for a given application.
If there is a never-treated cohort, i.e. $\infty\in supp\{E_{i}\}$,
then we may set $C=\{\infty\}$ and estimate regression (\ref{eq: saturated event studies model})
on all observations. If there are no never-treated units, i.e. $\infty\not\in supp\{E_{i}\}$,
then we may set $C=\{\max\{E_{i}\}\}$, i.e. the latest-treated cohort
and estimate regression (\ref{eq: saturated event studies model})
on observations from $t=0,\dots,\max\{E_{i}\}-1$. Lastly, if there
is a cohort that is always treated, i.e. $0\in supp\{E_{i}\}$, then
we need to exclude this cohort from estimation. 

The coefficient estimator $\widehat{\delta}_{e,\ell}$ from regression
(\ref{eq: saturated event studies model}) is a DID estimator for
$CATT_{e,\ell}$ with particular choices of pre-periods and control
cohorts. As DID is likely a familiar estimator for applied researchers,
we separate the more in-depth discussion of its definition, choices
of pre-periods, and choices of control cohorts in Section \ref{subsec:Difference-in-differences-estima}. 

\textbf{Step 2. }We estimate the weights $Pr\{E_{i}=e\mid E_{i}\in[-\ell,T-\ell]\}$
by sample shares of each cohort in the relevant period(s) $\ell\in g$. 

\textbf{Step 3.} To form our IW estimator, we take a weighted average
of estimates for $CATT_{e,\ell}$ from Step 1 with weight estimates
from step 2. More formally, the IW estimator is
\begin{equation}
\widehat{\nu}_{g}=\frac{1}{\left|g\right|}\sum_{\ell\in g}\sum_{e}\widehat{\delta}_{e,\ell}\widehat{Pr}\{E_{i}=e\mid E_{i}\in[-\ell,T-\ell]\}
\end{equation}
where $\widehat{\delta}_{e,\ell}$ is returned from step 1 and $\widehat{Pr}\{E_{i}=e\mid E_{i}\in[-\ell,T-\ell]\}$
is the estimated weight returned from step 2. We normalize the weights
further by the size of $g$. If $g$ is a singleton, then its size
is $\left|g\right|=1$.

\textbf{Validity of the IW estimator. }Under the parallel trends and
no anticipation assumptions the coefficient estimator $\widehat{\delta}_{e,\ell}$
from regression (\ref{eq: saturated event studies model}) is a consistent
estimator for $CATT_{e,\ell}$. The sample shares of each cohort are
also consistent estimators for the population shares. Thus, the IW
estimator is consistent for a weighted average of $CATT_{e,\ell}$
with weights equal to the share of each cohort in the relevant period(s). 

With a few standard assumptions (which we present as Assumption \ref{assu:(The-saturated-regression}
in Online Appendix \ref{sec:Proofs}) on regression (\ref{eq: saturated event studies model}),
we can show that each IW estimator is asymptotically normal and derive
its asymptotic variance. The large sample approximation allows us
to estimate the variance of IW estimators directly without relying
on bootstrapping as in \citet{callaway_santanna_ssrn2018}. However,
we only construct pointwise confidence interval valid for a given
IW estimator $\widehat{\nu}_{g}$. The bootstrap-based inference by
\citet{callaway_santanna_ssrn2018} constructs simultaneous confidence
intervals that are valid for the entire path of $\widehat{\nu}_{g}$.

\subsection{Difference-in-differences estimator for $CATT_{e,\ell}$\label{subsec:Difference-in-differences-estima}}
\begin{defn}
Assume cohort $e$ is non-empty i.e. $\sum_{i=1}^{N}\text{\ensuremath{\mathbf{1}}}\{E_{i}=e\}>0$.
Assume there exists some pre-period $s<e$ and some set of control
cohorts $C\subseteq\left\{ c:e+\ell<c\leq T\right\} $ that are non-empty
i.e. $\sum_{i=1}^{N}\text{\ensuremath{\mathbf{1}}}\{E_{i}\in C\}>0$.
Using the notion $\mathbb{E}_{N}$ to abbreviate the symbol $\frac{1}{N}\sum_{i=1}^{N}$,
the DID estimator with pre-period $s$ and control cohorts $C$ estimates
$CATT_{e,\ell}$ as 
\begin{align}
\hat{\delta}_{e,\ell}=\frac{\mathbb{E}_{N}[\left(Y_{i,e+\ell}-Y_{i,s}\right)\cdot\text{\ensuremath{\mathbf{1}}}\left\{ E_{i}=e\right\} ]}{\mathbb{E}{}_{N}[\text{\ensuremath{\mathbf{1}}}\left\{ E_{i}=e\right\} ]}-\frac{\mathbb{E}_{N}[\left(Y_{i,e+\ell}-Y_{i,s}\right)\cdot\text{\ensuremath{\mathbf{1}}}\left\{ E_{i}\in C\right\} ]}{\mathbb{E}_{N}[\text{\ensuremath{\mathbf{1}}}\left\{ E_{i}\in C\right\} ]}.\label{eq:DID-estimator}
\end{align}

The assumption of non-empty cohort $e$, existence of pre-period and
non-empty control cohorts makes the DID estimator well-defined. For
example, DID estimators for cohort 0 are not well-defined because
a pre-period does not exist for this cohort, which is why we exclude
them from estimating regression (\ref{eq: saturated event studies model})
in Step 1.

Using the above definition, in regression (\ref{eq: saturated event studies model})
from Step 1 of our proposed method, the coefficient estimator $\widehat{\delta}_{e,\ell}$
is a DID estimator for $CATT_{e,\ell}$ with pre-period $s=e-1$ (because
we exclude relative period $\ell=-1$) and some choice of control
cohorts $C$. If there is a never-treated cohort i.e. $\infty\in supp\{E_{i}\}$,
then we set the control cohort to be never-treated units $C=\{\infty\}$.
If there are no never-treated units, i.e. $\infty\not\in supp\{E_{i}\}$,
then we set the control cohort $C=\{\max\{E_{i}\}\}$, i.e. the latest-treated
cohort. Among all possible pre-periods, for regression (\ref{eq: saturated event studies model})
from step 1 of our proposed method we choose $s=e-1$ and $C=\{\infty\}$
with never-treated units (or $\{\max\{E_{i}\}\}$ without never-treated
units) because the resulting specification is a natural extension
of the common specifications of two-way fixed effects regression. 

Note that without never-treated units we need to drop time periods
$t\geq\max\{E_{i}\}$ from estimating regression (\ref{eq: saturated event studies model})
because every unit will be treated in these periods. The DID estimators
for $CATT_{e,\ell}$ for $e+\ell\geq\max\{E_{i}\}$ are thus not well-defined
as the control cohort is empty $C=\emptyset$. For example, when there
are just two cohorts with $E_{i}\in\{1,T\}$, one treated at $t=1$
and the other treated in the last period, we then need to omit interaction
terms involving the latest-treated cohorts as well as dropping observations
from the last period $t=T$ from estimation.
\end{defn}
Under some assumptions, this DID estimator $\widehat{\delta}_{e,\ell}$
is an unbiased and consistent estimator for $CATT_{e,\ell}$, a fact
that we build on in deriving the probability limit of the IW estimator.
We state this in the following proposition. 
\begin{prop}
\label{prop:The-DID-estimator} If Assumptions \ref{assum:PT-baseline}
and \ref{assum:no-anti} hold, then the DID estimator using any pre-period
$s<e$ and non-empty control cohorts $C$ is an unbiased and consistent
estimator for $CATT_{e,\ell}$.
\end{prop}
It is possible to relax the parallel trends assumption to allow the
timing of treatment to depend on covariates. One can estimate $CATT_{e,\ell}$
consistently based on the inverse propensity score reweighted estimator
proposed by \citet{semiparametric_did} and \citet{callaway_santanna_ssrn2018},
the outcome regression approach by \citet{heckman_ichimura_todd_restud1997},
and the doubly robust estimator recently proposed by \citet{santanna_zhao_dr_did2018}.
The resulting estimates $\hat{\delta}_{e,l}$ can then be plugged
into step 3 to form our IW estimator. In particular, without covariates
and for the case with never treated units, our approach coincides
with \citet{callaway_santanna_ssrn2018}\emph{.} Therefore one can
use the \texttt{did} R package developed by \citet{callaway_R} to
form the IW estimator.

The choice of the pre-period $s$ and the control cohorts $C$ depends
on the trade-off between relaxing Assumption \ref{assum:PT-baseline}
or Assumption \ref{assum:no-anti}. If Assumption \ref{assum:no-anti}
is likely to hold, we can choose any pre-period $s<e$ and include
not-yet-treated units in control cohorts $C=\{c:c>e+\ell\}$ as in
\citet{callaway_santanna_ssrn2018}. This choice allows us to relax
the parallel trends assumption to be just $E[Y_{i,e+\ell}^{\infty}-Y_{i,0}^{\infty}\vert E_{i}=e]=E[Y_{i,e+\ell}^{\infty}-Y_{i,0}^{\infty}\vert E_{i}>e+\ell]$.
If instead we want to relax Assumption \ref{assum:no-anti} to allow
some anticipation, we then need the parallel trends assumption to
hold. For example, suppose we are only willing to assume away anticipation
2 periods before a unit is treated, we may still be able to recover
several $CATT_{e,\ell}$ by appropriately selecting a pre-period and
a smaller set of control cohorts. We may choose any $s<e-2$ for the
pre-period, and for control cohorts, we may choose cohorts treated
only after at least 2 periods from now $C\subseteq\left\{ c:e+\ell+2<c\leq T\right\} $,
or choose the latest-treated cohort $C=\{\max\{E_{i}\}\}$ as specified
in regression (\ref{eq: saturated event studies model}).

\section{Empirical Illustration\label{sec:Applications}}

We illustrate our findings in the setting of \citet{dobkin_finkelstein_kluender_notowidigdo_aer2018}.
\citet{dobkin_finkelstein_kluender_notowidigdo_aer2018} study the
economic consequences of hospitalization, which is a large source
of economic risk for adults in the United States. To quantify these
economic risks, in the first part of their analysis, \citet{dobkin_finkelstein_kluender_notowidigdo_aer2018}
leverage variation in the timing of hospitalization observed in the
publicly-available dataset, Health and Retirement Study (HRS), which
we describe in more detail in Section \ref{subsec:Setting-and-data}.
Their estimation of the dynamic effects of hospitalization using two-way
fixed effects regressions provides a good context for demonstrating
our results. As we argue below in Section \ref{subsec:Assumptions},
parallel trends and no anticipation assumptions are plausible in this
setting. However, the effect of hospitalization is potentially heterogenous
across individuals hospitalized in different years. Our findings of
non-convex and non-zero weighting would therefore apply to their two-way
fixed effects regression estimates, and our alternative estimator
could lead to different estimates. Furthermore, this dataset is publicly
available, which allows us to provide replication files.

\subsection{Data\label{subsec:Setting-and-data}}

Our sample selection closely follows \citet{dobkin_finkelstein_kluender_notowidigdo_aer2018}
but we include a cursory explanation here for completeness with an
emphasis on how our final sample differs from their main analysis
sample. Our primary source of data is the biennial Health and Retirement
Study (HRS). We identify the sample of individuals who appear in two
sequential waves of surveys and newly report having a hospital admission
over the last two years (the ``index'' or initial admission) at
the second survey. To focus on health ``shocks'', we restrict attention
to non-pregnancy-related hospital admissions as in \citet{dobkin_finkelstein_kluender_notowidigdo_aer2018}.
We also follow \citet{dobkin_finkelstein_kluender_notowidigdo_aer2018}
by focusing on adults who are hospitalized at ages 50-59. 

Unlike \citet{dobkin_finkelstein_kluender_notowidigdo_aer2018}, we
restrict our analysis to a subsample of these individuals who appear
throughout waves 7-11 (roughly 2004-2012). Our sample of analysis
therefore includes HRS respondents with index hospitalization during
waves 8-11. The purpose of this sample restriction is to maintain
a balanced panel with a reasonable sample size. 

Here $i$ indexes an individual, and $t$ indexes survey wave ($T=4$)
and is normalized to zero for wave 7, the first wave in our sample.
Among the outcomes $Y_{i,t}$ studied by \citet{dobkin_finkelstein_kluender_notowidigdo_aer2018},
we focus on two: out-of-pocket medical spending and labor earnings.
They are derived from self-reports, adjusted to 2005 dollars and censored
at the 99.95th percentile.

\textbf{Summary statistics. }Table \ref{tab:Sample-characteristics}
presents basic summary statistics for our analysis sample before hospitalization.
We have a slightly lower fraction of white in our sample, but otherwise
have a similar sample to \citet{dobkin_finkelstein_kluender_notowidigdo_aer2018}. 

In Panel D, we compare means of the cross-sectional distributions
of outcomes for individuals who have not been hospitalized by each
wave. The size of the sample conditional on not having been hospitalized
strictly decreases with each subsequent wave. There are apparent time
trends in our outcomes of interest prior to hospitalization as we
observe distributional changes across waves. Out-of-pocket medical
spendings fluctuate and earnings decrease with each wave on average
as more individuals are retired in each subsequent wave.

\subsection{Setting\label{subsec:Assumptions}}

We illustrate how variations in the timing of hospitalization fit
the event studies design proposed in Section \ref{sec:Event-studies-in-a-potential}.
We define treatment $D_{i,t}$ to be ever having been hospitalized.
In our terminology, we categorize individuals into cohorts based on
$E_{i}$, which is defined as the survey wave of their initial hospitalization.
Since we restrict the sample to individuals who were ever hospitalized
in waves 8-11, there are four cohorts $E_{i}\in\left\{ 1,2,3,4\right\} $.
Although hospitalization itself may not be an absorbing state, we
are trying to model the impact of having had any hospitalization.
Thus, the cohort-specific average treatment effects $CATT_{e,\ell}$
trace out the path of treatment effects for cohort $e$ following
a negative health shock (even though the shock itself may be transient),
as opposed to never having been hospitalized. Next we discuss whether
each of the three identifying assumptions proposed is likely to hold
in the context of unexpected hospitalizations.

\textbf{Parallel trends (Assumption \ref{assum:PT-baseline}). }Hospitalization
is likely to be earlier among sicker individuals with high out-of-pocket
medical spending and low labor earnings. Thus, it is not plausible
that the baseline outcome $Y_{i,t}^{\infty}$ is mean independent
of the timing of hospitalization. The parallel trends assumption is
more plausible as it allows the timing to depend on unobserved time-invariant
characteristics such as chronic disease. Furthermore, hospitalized
individuals might be on a downward trend for labor earnings already
prior to hospitalization compared to individuals who are never hospitalized.
To reduce such confounding, we follow \citet{dobkin_finkelstein_kluender_notowidigdo_aer2018}
to restrict the parallel trends assumption to individuals who were
ever hospitalized. 

\textbf{No anticipatory behavior (Assumption \ref{assum:no-anti}).
}It is plausible that there is no anticipatory behavior prior to the
hospitalization, given that the treatment is restricted to conditions
that are likely unexpected hospitalizations. This assumption may be
violated if individuals have private information about the probability
of these hospitalizations over time and thus adjust their behavior
prior to hospitalization.

\textbf{Treatment effect heterogeneity (Assumptions \ref{assum:homogeneous-TE}).
}For out-of-pocket medical spending, the effect of hospitalization
is potentially heterogenous across individuals hospitalized in different
waves. Individuals hospitalized in later waves are mechanically older
at the time of hospitalization than individuals hospitalized in earlier
waves. The effect on out-of-pocket medical spending is largely determined
by generosity of health insurance, which may decrease as individuals
age into Medicare. The effect on labor earnings is also likely heterogenous
as it depends on the labor market condition at the time of hospitalization:
for example, individuals hospitalized during the financial crisis
may find it more difficult to return to the labor force, and suffer
a more grave decrease in earnings. 

\subsection{Illustrating weights in two-way fixed effects regression\label{subsec:Illustrating-weighting-in}}

We illustrate our results on two-way fixed effects regression by estimating
the following specification with indicators for up to three leads
and lags following equation (3) of \citet{dobkin_finkelstein_kluender_notowidigdo_aer2018}
\begin{equation}
Y_{i,t}=\alpha_{i}+\lambda_{t}+\mu_{-3}D_{i,t}^{-3}+\mu_{-2}D_{i,t}^{-2}+\mu_{0}D_{i,t}^{0}+\mu_{1}D_{i,t}^{1}+\mu_{2}D_{i,t}^{2}+\mu_{3}D_{i,t}^{3}+\upsilon_{i,t}.\label{eq: HRS-dyno}
\end{equation}
For their estimation, \citet{dobkin_finkelstein_kluender_notowidigdo_aer2018}
trim their sample, keeping only observations up to three waves prior
to the hospitalization and three waves after the hospitalization,
and weight their regression with survey weights. To fully illustrate
issues with this specification, we do not trim, but rather use a sample
balanced in calendar time for $t\in\{0,\dots,4\}$, and do not apply
survey weights. With a sample balanced in calendar time, we need to
exclude at least two relative period indicators due to multicollinearity.
Following \citet{dobkin_finkelstein_kluender_notowidigdo_aer2018},
we exclude the period right before hospitalization ($\ell=-1$). We
also exclude $\ell=-4$. Since we do not trim, note that our results
are not directly comparable to \citet{dobkin_finkelstein_kluender_notowidigdo_aer2018}
even though they are quite similar.

We focus on a single coefficient $\mu_{-2}$ that is supposed to test
for any pre-trend of hospitalizations. As in Proposition \ref{prop: dyno FE coeff weights},
we can decompose $\mu_{-2}$ as
\begin{equation}
\sum_{e=1}^{4}\omega_{e,0}^{-2}CATT_{e,0}+\sum_{\ell=-3,\neq-1}^{3}\sum_{e=1}^{4}\omega_{e,\ell}^{-2}CATT_{e,\ell}+\sum_{\ell'=-4,-1}\sum_{e=1}^{4}\omega_{e,\ell'}^{-2}CATT_{e,\ell'}.\label{eq:HRS-dyno-limit}
\end{equation}
As discussed in Proposition \ref{prop: dyno FE coeff weights-1} we
can estimate the underlying weights $\omega_{e,\ell}^{-2}$ by regressing
$\mathbf{1}\left\{ E_{i}=e\right\} \cdot D_{i,t}^{\ell}$ on the relative
wave indicators included in specification (\ref{eq: HRS-dyno}) i.e.
$\{D_{i,t}^{\ell}\}_{\ell=-3,\neq-1}^{3}$ and two-way fixed effects.
The coefficient estimator of $D_{i,t}^{-2}$ in such regression, $\widehat{\omega}_{e,\ell}^{-2}$,
consistently estimates $\omega_{e,\ell}^{-2}$. 

Figure \ref{fig:Weights} plots these estimated weights. As described
in our decomposition results, these weights have the following properties:
(a) the four weights from relative wave $\ell=-2$ sum to one; (b)
the weights for other included relative waves $\ell\in\{-3,0,1,2,3\}$
sum to zero for each included relative wave; and (c) the weights from
excluded relative waves $\ell\in\{-4,-1\}$ sum to negative one across
these excluded relative waves.

The weights are non-negative for lags of treatments, which suggest
that the FE estimate $\widehat{\mu}_{-2}$ is particularly sensitive
to estimates of the dynamic effects of hospitalizations and does not
isolate the pre-trends. Specifically, treatment effects heterogeneity
in $\ell=-3,0,1,2$ can affect the FE estimate $\widehat{\mu}_{-2}$.
Applied researchers can make similar plots to visualize the role of
weights in their settings with our publicly-available Stata package
{\footnotesize{}\texttt{eventstudyweights}} by \citet{sun_stata}.

\subsection{Comparing FE and IW estimates}

We illustrate our alternative method (IW estimator) following the
three steps outline in Section \ref{subsec:Interaction-weighted-estimator}.
First, we estimate the interacted specification (\ref{eq: saturated event studies model})
as
\begin{equation}
Y_{i,t}=\alpha_{i}+\lambda_{t}+\sum_{e\in\left\{ 1,2,3\right\} }\sum_{\ell=-3,\neq-1}^{2}\delta_{e,\ell}\mathbf{1}\{E_{i}=e\}\cdot D_{i,t}^{\ell}+\epsilon_{i,t}\label{eq: HRS-saturated model}
\end{equation}
for $t=0,\dots,3$. Specifically, we estimate $CATT_{e,\ell}$ using
a DID estimator $\widehat{\delta}_{e,\ell}$ with pre-period $s=-1$
and control cohort $C=4$, the cohort hospitalized in the last period.
This means we need to drop $t=4$ from estimation because everyone
has been hospitalized by $t=4$, and a control cohort for estimating
$CATT_{e,\ell}$ in $t=4$ does not exist. Second, we estimate the
sample share of each cohort $e$ across cohorts that experience at
least $l$ periods relative to hospitalization by its sample analog.
Third, we form IW estimates $\widehat{\nu}_{\ell}$ by taking weighted
averages of $\widehat{\delta}_{e,\ell}$ (returned from step one)
with sample cohort share (returned from step two) as weights.

In Table \ref{tab:Estimates-for-HRS}, we report the FE estimates
$\widehat{\mu}_{\ell}$ and the IW estimates $\widehat{\nu}_{\ell}$,
as well as the underlying $CATT_{e,\ell}$ estimates $\widehat{\delta}_{e,\ell}$.
In this application, the FE estimates and IW estimates happen to be
very similar in their magnitude. The conclusion from \citet{dobkin_finkelstein_kluender_notowidigdo_aer2018}
based on FE estimates similar to ours still holds: we find a substantial
and persistent decline in the earnings due to hospitalization, and
the increase in out-of-pocket spending is transitory and small in
comparison. 

However, the FE estimates still suffer from non-convex and non-zero
weighting as we show in Section \ref{subsec:Illustrating-weighting-in},
which can lead to interpretability issues depending on the amount
of treatment effect heterogeneity. Even though for both outcomes $\widehat{\mu}_{0}$
falls in the convex hull of its underlying $CATT_{e,0}$ estimates,
$\widehat{\mu}_{-2}$ turns out to be outside the convex hull of its
underlying $CATT_{e,-2}$ estimates. In contrast, by construction,
the IW estimates $\widehat{\nu}_{\ell}$ fall within the convex hull
of its underlying $CATT_{e,\ell}$ estimates and are unaffected by
$CATT_{e,\ell'}$ estimates from other periods $\ell'\neq\ell$. Thus,
they have an interpretation as an average effect of the treatment
on the treated $\ell$ periods after initial treatment. 

\section{Conclusions\label{sec:Conclusions}}

This paper analyzes the behavior of relative period coefficients $\mu_{\ell}$
on the indicator for being $\ell$ periods away from the treatment
from two-way fixed effects regressions in settings with variation
in treatment timing and treatment effects heterogeneity. For dynamic
treatment effects, researchers are usually interested in estimating
some average of treatment effects from $\ell$ periods relative to
the treatment, and it is common to report the coefficient estimate
$\widehat{\mu}_{\ell}$ assuming that interpretation is valid. However,
we show that in the presence of heterogenous treatment effects, the
coefficient $\mu_{\ell}$ does not necessarily capture the dynamic
treatment effect as it can fall outside the convex hull of $CATT_{e,\ell}$
from its corresponding period $\ell$ and may pick up spurious terms
consisting of treatment effects from periods other than $\ell$. This
negative result is based on the decomposition of the coefficients
$\mu_{\ell}$ as a linear combination of cohort-specific average treatment
effects on the treated $CATT_{e,\ell}$. The weights in the linear
combination can be non-convex, and non-zero on relative periods $\ell'\neq\ell$. 

Given these negative results on two-way fixed effects regression estimators,
we propose ``interaction-weighted'' (IW) estimators for estimating
dynamic treatment effects. The IW estimators are formed by first estimating
$CATT_{e,\ell}$ with a regression saturated in cohort and relative
period indicators, and then averaging estimates of $CATT_{e,\ell}$
across $e$ at a given $\ell$. These $CATT_{e,\ell}$ are identified
under parallel trends and no anticipation assumptions. These estimators
are easy to implement and robust to heterogenous treatment effects
across cohorts; the IW estimator associated with relative period $\ell$
is guaranteed to estimate a convex average of $CATT_{e,\ell}$ using
weights that are sample share of each cohort $e$.

Finally, we illustrate the empirical relevance of our results by estimating
the dynamic effect of hospitalization on the out-of-pocket medical
spending and labor earnings using a setup similar to \citet{dobkin_finkelstein_kluender_notowidigdo_aer2018}.
We find non-convex and non-zero weights for two-way fixed effects
regression in this example, and show that the resulting estimates
indeed sometimes fall outside the convex hull of the underlying \emph{CATT
}estimates due to contamination by treatment effects from other relative
periods. IW estimates, on the other hand, are weighted averages of
the underlying\emph{ CATT }estimates with weights representative of
cohort share.

More broadly, our paper suggests that researchers can do more in this
context to assess the validity of underlying assumptions and how violations
may impact their estimates. We demonstrate the sensitivity of two-way
fixed effects regressions to underlying assumptions and provide researchers
with tools to assess and address these issues. Specifically, for average
treatment effects estimated using two-way fixed effects regressions
we recommend that empirical researchers directly estimate the underlying
weights on cohort-specific average treatment effects using our proposed
auxiliary regressions. This exercise allows empirical researchers
to assess the degree of potential contamination under any assumed
structure of treatment effects heterogeneity. We also recommend empirical
researchers consider estimation methods more robust to treatment effects
heterogeneity such as the IW estimator we propose. Developing additional
tools for applied researchers to identify and estimate interpretable
parameters of interest is a promising avenue for future research.

\newpage{}

\bibliographystyle{aer}
\bibliography{mover}
\newpage{}

\include{event_studies_figs_tables}

\appendix
\textbf{Online Appendix}

\input{event_studies_appendix.tex}

\end{document}

%% file: event_studies_figs_tables.tex
\begin{figure}[H]
\begin{centering}
\caption{\label{fig:Weights-on-mu_m2}Weight on $CATT_{e,\ell'}$ for $\ell'\protect\neq-2$
in the Coefficient $\mu_{-2}$ from Regression (\ref{eq:weights ex rel times-1})}
\subfloat[Weight on $CATT_{e,0}$ ]{
\centering{}\includegraphics[scale=0.8]{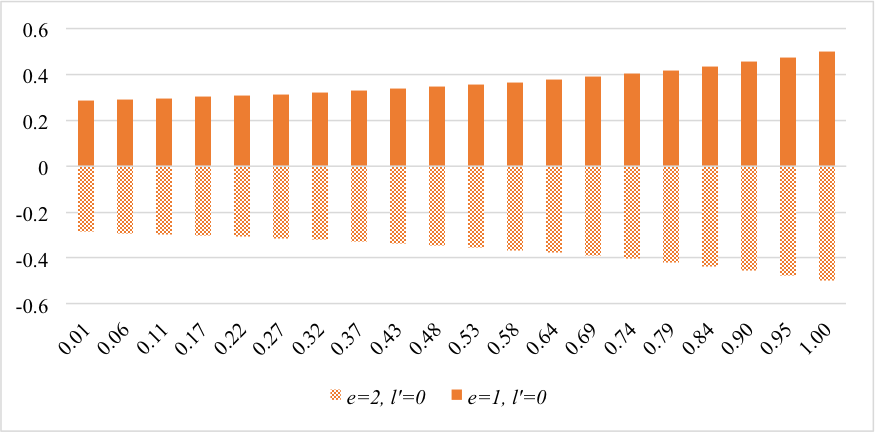}}\\
\subfloat[Weight on $CATT_{e,\ell'}$ for $\ell'\in\{-1,1\}$]{\centering{}\includegraphics[scale=0.8]{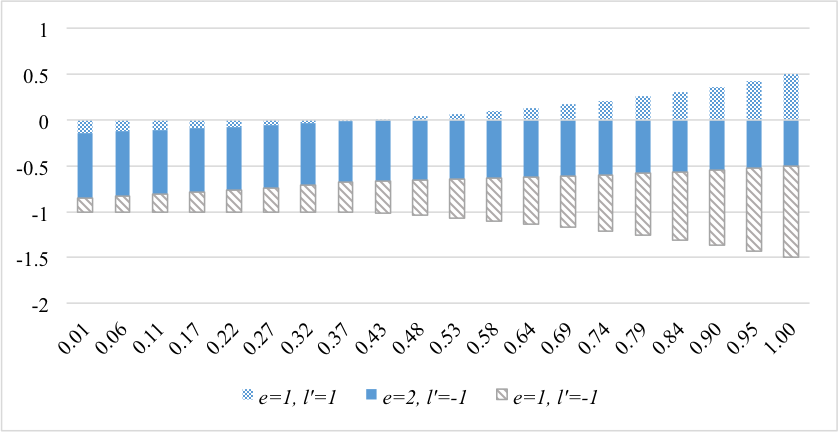}}
\par\end{centering}
\emph{Notes:} Panel (a) plots the weight on $CATT_{e,0}$ and panel
(b) plots the weight on $CATT_{e,\ell'}$ for $\ell'\in\{-1,1\}$
under a given distribution of cohorts indexed by the $x$-axis. Specifically,
we construct a panel balanced in calendar time with $T=2$ and cohorts
$E_{i}\in\{1,2,\infty\}$. We vary the share of never-treated units
$Pr\{E_{i}=\infty\}$ between $[0,0.99]$ and divide the rest of the
units evenly into cohorts treated at time 1 and 2. The $x$-axis indexes
$1-Pr\{E_{i}=\infty\}$.
\end{figure}

\pagebreak{}
\begin{figure}[H]
\begin{centering}
\caption{FE\emph{ }vs IW Estimates of the Effects of Hospitalization on Outcomes\label{fig:FEvsIW}}
\par\end{centering}
\begin{centering}
\subfloat[Out-of-pocket Medical Spending]{
\centering{}\includegraphics{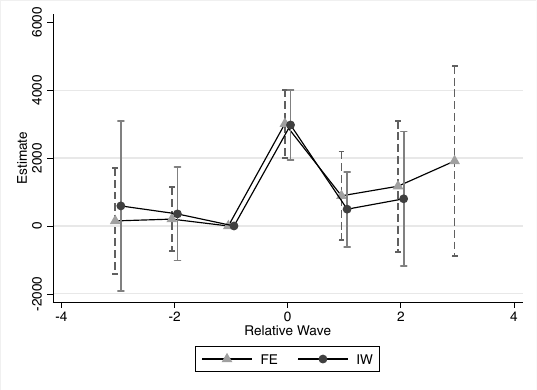}}\\
\subfloat[Labor Earnings]{
\centering{}\includegraphics{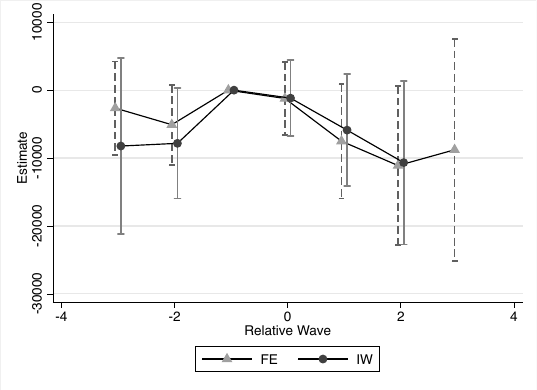}}
\par\end{centering}
\emph{Notes:} Each figure plots FE estimates $\widehat{\mu}_{\ell}$
from regression (\ref{eq: HRS-dyno}) in triangles and IW estimates
$\widehat{\nu}_{\ell}$ from regression (\ref{eq: HRS-saturated model})
in circles against relative wave $\ell$, with their respective pointwise
95\% confidence intervals. Both are estimates for the effect of hospitalization
at relative wave $\ell$. The outcome variable is out-of-pocket medical
spending in panel (a) and labor earnings in panel (b) respectively.
\end{figure}
\pagebreak{}
\begin{figure}[H]
\begin{centering}
\caption{\label{fig:Weights} Estimated Weights $\widehat{\omega}_{e,\ell}^{-2}$
Underlying $\mu_{-2}$ }
\par\end{centering}
\begin{centering}
\includegraphics[scale=1.5]{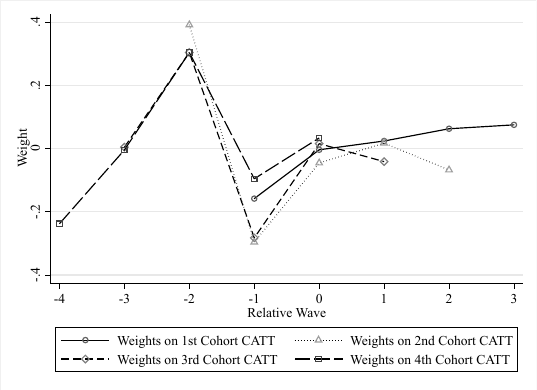}
\par\end{centering}
\emph{Notes: }The FE estimator for the pretrend of hospitalization
$\widehat{\mu}_{-2}$ estimates a linear combination of cohort-specific
effects $CATT_{e,\ell}$'s from all cohorts $e$ and relative waves
$\ell$ as specified in equation (\ref{eq:HRS-dyno-limit}). This
figure plots the estimated weight $\widehat{\omega}_{e,\ell}^{-2}$
associated with each $CATT_{e,\ell}$ in this linear combination. 
\end{figure}
\pagebreak{}\clearpage
\begin{sidewaystable}
\begin{centering}
\caption{\label{tab:Selected-publications}Survey of Applied Papers}
\begin{tabular}{>{\raggedright}p{0.2\textwidth}>{\centering}p{0.1\textwidth}>{\centering}p{0.1\textwidth}>{\centering}p{0.1\textwidth}>{\centering}p{0.1\textwidth}>{\centering}p{0.1\textwidth}>{\centering}p{0.1\textwidth}>{\centering}p{0.1\textwidth}}
\hline 
Paper & Binary Absorbing Treatment & Variation in Treatment Timing & Pre-Treatment Relative Period Excluded & Exclude Distant Relative Periods & Bin Distant Relative Periods & Includes Never Treated Units & Panel Balanced in Relative Time\tabularnewline
\hline 
Bosch and Campos-Vazquez (2014) & X &  &  &  &  &  & \tabularnewline
\hline 
Fitzpatrick and Lovenheim (2014) &  &  &  &  &  &  & \tabularnewline
\hline 
Gallagher (2014) & X & X & -1 &  & X & X & \tabularnewline
\hline 
Tewari (2014) & X & X & 0 &  &  &  & X\tabularnewline
\hline 
Ujhelyi (2014) & X & X & -1 &  & X & X & \tabularnewline
\hline 
Bailey and Goodman-Bacon (2015) & X & X & -1 &  & X & X & X\tabularnewline
\hline 
Deryugina (2017) & X & X & -1 &  & X & X & X\tabularnewline
\hline 
Deschenes et al. (2017) & X &  &  &  &  &  & \tabularnewline
\hline 
He and Wang (2017) & X & X & -1 &  & X &  & \tabularnewline
\hline 
Lafortune et al. (2017) & X & X & 0 & X &  & X & \tabularnewline
\hline 
Kuziemko et al. (2018 & X & X & -1 &  & X &  & \tabularnewline
\hline 
Markevich and Zhuravskaya (2018) &  &  &  &  &  &  & \tabularnewline
\hline 
\end{tabular}
\par\end{centering}
\emph{Notes: }This table consolidates key properties of main event
study specifications across a sample of applied papers. We follow
the selection criteria in \cite{roth_pretest_2019}: the original
sample consists of 70 total papers, but is further constrained to
these twelve papers with publicly available data and code. The data
and code are used to determine exactly the specification estimated
in these papers. We focus on the first specification underlying the
event study estimates in each paper, which we view as a reasonable
proxy for the main specification in the paper. Note that two papers
(Fitzpatrick and Lovenheim (2014); Markevich and Zhuravskaya (2018))
have none of the attributes listed in the columns.
\end{sidewaystable}
\clearpage\pagebreak{}
\begin{table}[H]
\begin{centering}
\caption{\label{tab:Sample-characteristics}Summary Statistics of the HRS Sample}
\par\end{centering}
\begin{centering}
\begin{tabular}{lrrr}
\hline 
 & $N$ & Mean & Std. Dev\tabularnewline
\hline 
\hline 
\emph{Panel A. Demographics} &  &  & \tabularnewline
Age at admission & 656 &  56 &  2.29\tabularnewline
Male & 656 &  0.456 &  0.498\tabularnewline
Year of admission & 656 &  2,007 &  2.11\tabularnewline
 &  &  & \tabularnewline
\emph{Panel B. Race/ethnicity} &  &  & \tabularnewline
Hispanic & 656 & 0.122 & 0.327\tabularnewline
Black & 656 & 0.151 & 0.358\tabularnewline
White & 656 & 0.742 & 0.438\tabularnewline
Other race & 656 & 0.107 & 0.309\tabularnewline
 &  &  & \tabularnewline
\emph{Panel C. Insurance status} &  &  & \tabularnewline
Medicaid & 656 & 0.05 & 0.219\tabularnewline
Private & 656 & 0.715 & 0.452\tabularnewline
Medicare & 656 & 0.072 & 0.259\tabularnewline
 &  &  & \tabularnewline
\emph{Panel D. Pre-hospitalization outcome} &  &  & \tabularnewline
Out-of-pocket medical spending &  &  & \tabularnewline
Wave 7 ($E_{i}\geq1$) & 656 & 3,302 & 9,024\tabularnewline
Wave 8 ($E_{i}\geq2$) & 404 & 2,355 & 8,132\tabularnewline
Wave 9 ($E_{i}\geq3$) &  228 & 2,056 &  3,532\tabularnewline
Wave 10 ($E_{i}=4$) &  65 & 2,044 & 4,379\tabularnewline
 &  &  & \tabularnewline
Earnings &  &  & \tabularnewline
Wave 7 ($E_{i}\geq1$) & 656 & 43,810 & 67,950\tabularnewline
Wave 8 ($E_{i}\geq2$) & 404 & 38,944 & 58,601\tabularnewline
Wave 9 ($E_{i}\geq3$) &  228 & 36,274 & 56,768\tabularnewline
Wave 10 ($E_{i}=4$) &  65 & 29,037 & 46,289\tabularnewline
\hline 
\end{tabular}
\par\end{centering}
\emph{Notes: }This table presents summary statistics on our primary
analysis sample, taken from the biennial Health and Retirement Survey
(HRS). We include the sample of individuals ages 50-59 in waves 7-11
(approximately spanning 2004-2012) who appear in two sequential survey
waves and report a recent hospital admission in the second survey.
For panel D, the sample corresponding to wave $t$ is conditional
on not having hospitalization by wave $t$.
\end{table}

\pagebreak{}

\begin{table}[H]
\caption{\label{tab:Estimates-for-HRS}Estimates for the Effect of Hospitalization
on Outcomes}
 
\begin{centering}
\subfloat[Out-of-pocket Medical Spending]{\centering{}%
\begin{tabular}{c>{\centering}p{1.5cm}>{\centering}p{1.5cm}>{\centering}p{1.5cm}>{\centering}p{1.5cm}>{\centering}p{1.5cm}}
\hline 
 & FE  & IW  & \multicolumn{3}{c}{$CATT_{e,\ell}$}\tabularnewline
$\ell$ Wave Relative to Hospitalization & $\widehat{\mu}_{\ell}$ & $\widehat{\nu}_{\ell}$ & $\widehat{\delta}_{1,\ell}$ & $\widehat{\delta}_{2,\ell}$ & $\widehat{\delta}_{3,\ell}$\tabularnewline
\hline 
-3 &  149  &  591  &  -  &  -  &  591 \tabularnewline
 &  (792)  &  (1273)  &  &  &  (1273) \tabularnewline
-2 &  203  &  353  &  -  &  299  &  411 \tabularnewline
 &  (480)  &  (698)  &  &  (967)  &  (1030) \tabularnewline
-1 & 0 & 0 & 0 & 0 & 0\tabularnewline
 &  &  &  &  & \tabularnewline
0 &  3,013  &  2,960  &  2,826  &  3,031  &  3,092 \tabularnewline
 &  (511)  &  (543)  &  (1038)  &  (704)  &  (998) \tabularnewline
1 &  888  & 530  &  825  &  107  &  - \tabularnewline
 &  (664)  &  (587)  &  (912)  &  (653)  & \tabularnewline
2 &  1,172  &  800  &  800  &  -  &  - \tabularnewline
 &  (983)  &  (1010)  &  (1010)  &  & \tabularnewline
3 &  1,914  &  -  &  -  &  -  &  - \tabularnewline
 &  (1426)  &  &  &  & \tabularnewline
\hline 
\end{tabular}}\\
\subfloat[Labor Earnings]{\begin{centering}
\begin{tabular}{c>{\centering}p{1.5cm}>{\centering}p{1.5cm}>{\centering}p{1.5cm}>{\centering}p{1.5cm}>{\centering}p{1.5cm}}
\hline 
 & FE & IW & \multicolumn{3}{c}{$CATT_{e,\ell}$}\tabularnewline
$\ell$ Wave Relative to Hospitalization & $\widehat{\mu}_{\ell}$ & $\widehat{\nu}_{\ell}$ & $\widehat{\delta}_{1,\ell}$ & $\widehat{\delta}_{2,\ell}$ & $\widehat{\delta}_{3,\ell}$\tabularnewline
\hline 
-3 & -2,642  & -8,228  &  -  &  -  & -8,228 \tabularnewline
 &  (3504)  &  (6592)  &  &  &  (6592) \tabularnewline
-2 & -5,089  & -7,823  &  -  & -7,691  & -7,964 \tabularnewline
 &  (3005)  &  (4197)  &  &  (6348)  &  (4766) \tabularnewline
-1 & 0 & 0 & 0 & 0 & 0\tabularnewline
 &  &  &  &     &   \tabularnewline
0 & -1,225  & -863  &  455  &  3,032  & -7,107 \tabularnewline
 &  (2742)  &  (2967)  &  (5593)  &  (4152)  &  (5883) \tabularnewline
1 & -7,508  & -5,435 & -1,670  & -10,826  &  - \tabularnewline
 &  (4312)  &  (4348)  &  (6500)  &  (4863)  & \tabularnewline
2 & -11,102  & -10,670  & -10,670  &  -  &  - \tabularnewline
 &  (5976)  &  (6154)  &  (6154)  &  & \tabularnewline
3 & -8,780  &  -  &  -  &  -  &  - \tabularnewline
 &  (8332)  &  &  &  & \tabularnewline
\hline 
\end{tabular}
\par\end{centering}
}
\par\end{centering}
\emph{Notes:} This table reports three different sets of estimates
for the dynamic effects of hospitalization on out-of-pocket medical
spending in panel (a) and labor earnings in panel (b). The first column
reports estimates from the FE estimator $\widehat{\mu}_{\ell}$. The
sample includes observations from wave $t=0,1,2,3,4$. Columns 3-5
report the estimates for $CATT_{e,\ell}$ from $\widehat{\delta}_{e,\ell}$.
The sample includes observations from wave $t=0,1,2,3$. Column 2
reports the IW estimates which are constructed as the weighted average
of $\widehat{\delta}_{e,\ell}$'s across cohorts $e$ who are $\ell$
periods from hospitalization. Standard errors (clustered on the individual)
are shown in parentheses.
\end{table}

%% file: event_studies_appendix.tex
\setlength{\abovedisplayskip}{0pt} \setlength{\belowdisplayskip}{0pt}

\section{Notation glossary and properties of double-demeaning}

In this section, we provide relevant notations for proofs in Section
\ref{sec:Proofs}. Section \ref{subsec:Expressions-weights} introduces
detailed expression for the weight.

We observe a balanced panel of $N$ i.i.d. observations $\{\{Y_{i,t},D_{i,t}\}_{t=0}^{T}\}_{i=1}^{N}$
where $Y_{i,t}\in\mathbb{R}$ is a real valued outcome variable and
$D_{i,t}\in\left\{ 0,1\right\} $ is a binary absorbing treatment
status variable: $D_{i,t}=0$ if $i$ is not treated in period $t$
and $D_{i,t}=1$ if $i$ is treated in period $t$. 

Since the treatment is absorbing, we can aggregate the treatment path
$\{D_{i,t}\}_{t=0}^{T}$ to a single discrete random variable $E_{i}=\min\{t:D_{i,t}=1\}$,
which is the period of the initial treatment. Additionally, we define
$D_{i,t}^{\ell}\coloneqq\mathbf{1}\{t-E_{i}=\ell\}$ to be an indicator
for being $\ell$ time periods relative to unit $i$'s initial treatment. 

We focus on the following two-way fixed effects regression

\[
Y_{i,t}=\alpha_{i}+\lambda_{t}+\sum_{g\in\mathcal{G}}\mu_{g}\mathbf{1}\{t-E_{i}\in g\}+\upsilon_{i,t}
\]
where $g$ are disjoint sets of relative times, $\alpha_{i}$ and
$\lambda_{t}$ are unit and time fixed effects. 
\begin{defn}
\label{def:double demean}For fixed $T$, consider a random vector
$\{X_{i,t}\}_{t=0}^{T}\in\mathbb{R}^{T+1}$, jointly distributed according
to $P$. At each $t$, let $\ddot{X}_{i,t}$ denote the following
random variable:
\begin{equation}
\ddot{X}_{i,t}=X_{i,t}-E[X_{i,t}]-\frac{1}{T+1}\sum_{s=0}^{T}X_{i,s}+\frac{1}{T+1}\sum_{s=0}^{T}E[X_{i,s}]
\end{equation}
The expectation $E[X_{i,t}]$ is taken cross-sectionally with respect
to $P$ at time $t=0,\dots,T$. 
\end{defn}
\begin{rem}
Conventionally $\ddot{X}_{i,t}$ is referred to as the double-demeaned
version of $X_{i,t}$, as it removes the contemporaneous expectation
$E[X_{i,t}]$ as well as the simple average across time for unit $i$,
$\frac{1}{T+1}\sum_{s=0}^{T}X_{i,t}$. It adds back a simple average
across time of the expectations so that $\ddot{X}_{i,t}$ has the
zero-mean and zero-sum properties as detailed in Lemma \ref{lem:Properties-of-double-demea}.
\end{rem}
\begin{lem}
\emph{\label{lem:Properties-of-double-demea}(Properties of double-demeaning.)}
For any $\ddot{X}_{i,t}$ and $\ddot{Z}_{i,t}$, double-demeaned versions
of $X_{i,t}$ and $Z_{i,t}$ respectively, we have the following properties:
\end{lem}
\begin{itemize}
\item zero-mean and zero-sum: $E[\ddot{X}_{i,t}]=0$ and $\sum_{t=0}^{T}\ddot{X}_{i,t}=0$;
\item idempotent: $\sum_{t=0}^{T}E[\ddot{X}_{i,t}\ddot{Z}_{i,t}]=\sum_{t=0}^{T}E[X_{i,t}\ddot{Z}_{i,t}]=\sum_{t=0}^{T}E[\ddot{X}_{i,t}Z_{i,t}]$;
\item for any time-invariant random variables $Z_{i,t}$ such that $Z_{i,t}=Z_{i}$,
double-demeaning annihilates it: $\ddot{Z}_{i,t}=0$.
\end{itemize}
\begin{proof}
The zero-mean and zero-sum properties hold by definition of double-demeaning:
\begin{align}
E[\ddot{X}_{i,t}] & =E\left[X_{i,t}-E[X_{i,t}]-\frac{1}{T+1}\sum_{s=0}^{T}X_{i,s}+\frac{1}{T+1}\sum_{s=0}^{T}E[X_{i,s}]\right]\\
 & =E[X_{i,t}]-E[X_{i,t}]-\frac{1}{T+1}\sum_{s=0}^{T}E[X_{i,s}]+\frac{1}{T+1}\sum_{s=0}^{T}E[X_{i,s}]=0
\end{align}
and 
\begin{align}
\sum_{t=0}^{T}\ddot{X}_{i,t} & =\sum_{t=0}^{T}X_{i,t}-\sum_{t=0}^{T}E[X_{i,t}]-\sum_{t=0}^{T}\frac{1}{T+1}\sum_{s=0}^{T}X_{i,s}+\sum_{t=0}^{T}\frac{1}{T+1}\sum_{s=0}^{T}E[X_{i,s}]\\
 & =\sum_{t=0}^{T}X_{i,t}-\sum_{t=0}^{T}E[X_{i,t}]-\sum_{s=0}^{T}X_{i,s}+\sum_{s=0}^{T}E[X_{i,s}]=0
\end{align}
For the idempotent property, first note that by definition of double-demeaning,
rearranging terms we can write $E[\ddot{X}_{i,t}\ddot{Z}_{i,t}]$
as
\begin{align}
 & =E\left[\left(X_{i,t}-\frac{1}{T+1}\sum_{s=0}^{T}X_{i,s}\right)\ddot{Z}_{i,t}\right]-E\left[\left(E[X_{i,t}]-\frac{1}{T+1}\sum_{s=0}^{T}E[X_{i,s}]\right)\ddot{Z}_{i,t}\right]\\
 & =E\left[\left(X_{i,t}-\frac{1}{T+1}\sum_{s=0}^{T}X_{i,s}\right)\ddot{Z}_{i,t}\right]-\left(E[X_{i,t}]-\frac{1}{T+1}\sum_{s=0}^{T}E[X_{i,s}]\right)\cdot E[\ddot{Z}_{i,t}]
\end{align}
By the zero-mean property $E[\ddot{Z}_{i,t}]=0$, the second term
of the above expression is zero. Summing the first term over $t$,
we have 
\begin{align}
 & \sum_{t=0}^{T}E[X_{i,t}\ddot{Z}_{i,t}]-\sum_{t=0}^{T}E\left[\left(\frac{1}{T+1}\sum_{s=0}^{T}X_{i,s}\right)\ddot{Z}_{i,t}\right]=\sum_{t=0}^{T}E[X_{i,t}\ddot{Z}_{i,t}]-E[\left(\frac{1}{T+1}\sum_{s=0}^{T}X_{i,s}\right)\sum_{t=0}^{T}\ddot{Z}_{i,t}]=\sum_{t=0}^{T}E[X_{i,t}\ddot{Z}_{i,t}]
\end{align}
where the last equality follows from the zero-sum property $\sum_{t=0}^{T}\ddot{Z}_{i,t}=0$.
This proves $\sum_{t=0}^{T}E[\ddot{X}_{i,t}\ddot{Z}_{i,t}]=\sum_{t=0}^{T}E[X_{i,t}\ddot{Z}_{i,t}]$.
Similarly we can show $\sum_{t=0}^{T}E[\ddot{X}_{i,t}\ddot{Z}_{i,t}]=\sum_{t=0}^{T}E[\ddot{X}_{i,t}Z_{i,t}]$. 

The annihilating property of double-meaning holds by definition for
any time-invariant random variable. Specifically, plugging in $Z_{i,t}=Z_{i}$
we have 
\begin{align}
\ddot{Z}_{i,t} & =Z_{i,t}-E[Z_{i,t}]-\frac{1}{T+1}\sum_{s=0}^{T}Z_{i,s}+\frac{1}{T+1}\sum_{s=0}^{T}E[Z_{i,s}]\\
 & =Z_{i}-E[Z_{i}]-Z_{i}+E[Z_{i}]=0
\end{align}
\end{proof}

\subsection{Expression of the weight\label{subsec:Expressions-weights}}

In this section, we provide detailed definitions for the three elements
in the expression of the weight $\omega_{e,\ell'}^{\ell}=\sigma_{e,\cdot}'\mathbf{\Delta}_{e+\ell'}A_{\ell}^{-1}$.
Recall $\omega_{e,\ell'}^{\ell}$ is the weight associated with $CATT_{e,\ell'}$
in the linear combination underlying $\mu_{\ell}$:
\begin{align*}
\omega_{e,\ell'}^{\ell} & =\underbrace{\mathbf{e}_{\ell}^{\intercal}\left(\sum_{t=0}^{T}E[\ddot{\mathbf{D}}_{i,t}\mathbf{D}_{i,t}^{\intercal}]\right)^{-1}}_{\coloneqq(A_{\ell}^{-1})^{\intercal}}E[\ddot{\mathbf{D}}_{i,t}D_{i,t}^{\ell'}\cdot\mathbf{1}\left\{ E_{i}=e\right\} ]\\
 & =(A_{\ell}^{-1})^{\intercal}\sum_{t=0}^{T}Cov(\mathbf{D}_{i,t}-\frac{1}{T+1}\sum_{t=0}^{T}\mathbf{D}_{i,t},D_{i,t}^{\ell'}\cdot\mathbf{1}\left\{ E_{i}=e\right\} -\frac{1}{T+1}\sum_{t=0}^{T}D_{i,t}^{\ell'}\cdot\mathbf{1}\left\{ E_{i}=e\right\} )\\
 & =(A_{\ell}^{-1})^{\intercal}Pr\left\{ E_{i}=e\right\} \cdot\Big(\big(E[\mathbf{D}_{i,e+\ell'}-\frac{1}{T+1}\sum_{t=0}^{T}\mathbf{D}_{i,e+\ell'}\mid E_{i}=e]\big)-\big(E[\mathbf{D}_{i,e+\ell'}]-E[\frac{1}{T+1}\sum_{t=0}^{T}\mathbf{D}_{i,e+\ell'}]\big)\Big)\\
 & =(A_{\ell}^{-1})^{\intercal}\sum_{e'}E[\mathbf{D}_{i,e+\ell'}-\frac{1}{T+1}\sum_{t}\mathbf{D}_{i,t}\mid E_{i}=e']\cdot Cov\left(\mathbf{1}\{E_{i}=e\},\mathbf{1}\{E_{i}=e'\}\right)\\
 & \coloneqq(\sigma_{e,\cdot})^{\intercal}\mathbf{\Delta}_{e+\ell'}A_{\ell}^{-1}
\end{align*}
We note that 
\begin{itemize}
\item The first term is the vector $\sigma_{e,\cdot}=(\sigma_{e,e'})_{e'\in supp(E_{i})}$
of the covariance between cohort indicators. Each entry $\sigma_{e,e'}=Cov\left(\mathbf{1}\{E_{i}=e\},\mathbf{1}\{E_{i}=e'\}\right)$
is the covariance between being in cohort $e$ and being in cohort
$e'$.
\item The second term is the matrix $\mathbf{\Delta}_{e+\ell'}$ of demeaned
relative time indicators. The rows of the matrix are indexed by cohorts
$e'\in supp(E_{i})$ and the columns of the matrix are indexed by
relative time indicators included in the specification (\ref{eq:dynamic FE}).
The entry that corresponds to cohort $e'$ and relative time indicator
$D_{i,t}^{\ell}$ is $E[D_{i,e+\ell'}^{\ell}-\frac{1}{T+1}\sum_{t}D_{i,t}^{\ell}\mid E_{i}=e']$.
Note that the sum $\sum_{t}D_{i,t}^{\ell}$ is equal to $\mathbf{1}\left\{ E_{i}\in\mathcal{I}_{\ell}\right\} $
the indicator for whether this individual ever experience the relative
time $\ell$.
\item The third term is the row vector of $A^{-1}$ that corresponds to
the relative time indicator $D_{i,t}^{\ell}$. Here $A$ is the covariance
matrix of demeaned relative time indicators included in the specification
(\ref{eq:dynamic FE}). Specifically, the entry that corresponds to
the covariance between demeaned $D_{i,t}^{\ell}$ and $D_{i,t}^{\ell'}$
is 
\[
A=\sum_{t}Cov\left(D_{i,t}^{\ell}-\frac{1}{T+1}\sum_{t}D_{i,t}^{\ell},D_{i,t}^{\ell'}-\frac{1}{T+1}\sum_{t}D_{i,t}^{\ell'}\right).
\]
Plugging in $\sum_{t}D_{i,t}^{\ell}=\mathbf{1}\left\{ E_{i}\in\mathcal{I}_{\ell}\right\} $
gives the expression in the main text.
\end{itemize}

\section{Proofs of decomposition results\label{sec:Proofs}}

\subsection*{Proof of Proposition \ref{prop: dyno FE coeff weights-1}}
\begin{proof}
Collect the relative time indicators in a column vector $\mathbf{D}_{i,t}=\left(\mathbf{1}\{t-E_{i}\in g\}\right){}_{g\in\mathcal{G}}^{\intercal}$.
Collect their corresponding coefficients in a column vector $\mu=(\mu_{g})_{g\in\mathcal{G}}^{\intercal}$.
Partialling out the unit and time fixed effects, Regression (\ref{eq:dynamic})
is 
\begin{equation}
\ddot{Y}_{i,t}=\mu^{\intercal}\ddot{\mathbf{D}}_{i,t}+\upsilon_{i,t}
\end{equation}
where $\ddot{X}_{i,t}$ is time- and cross-sectional demeaned version
of $X_{i,t}$ as defined in Definition \ref{def:double demean}. By
the idempotent property of Lemma \ref{lem:Properties-of-double-demea}
\begin{equation}
\mu_{g}=\mathbf{e}_{g}^{\intercal}\left(\sum_{t=0}^{T}E[\ddot{\mathbf{D}}_{i,t}\mathbf{D}_{i,t}^{\intercal}]\right)^{-1}\sum_{t=0}^{T}E[\ddot{\mathbf{D}}_{i,t}Y_{i,t}]\label{eq:dyno-ols}
\end{equation}
where $\mathbf{e}_{g}$ is a column vector with one in the entry corresponding
to the entry of $\mathbf{1}\{t-E_{i}\in g\}$ in $\mathbf{D}_{i,t}$,
and zero otherwise. 

To further develop the expression for the population regression coefficients
$\mu_{g}$, we note that by Lemma \ref{lem:Properties-of-double-demea},
we have $\sum_{t=0}^{T}E[\ddot{\mathbf{D}}_{i,t}Y_{i,0}^{\infty}]=\sum_{t=0}^{T}E[\mathbf{D}_{i,t}\ddot{Y}_{i,0}^{\infty}]=0$
since $Y_{i,0}^{\infty}$ is time-invariant, and $E[\ddot{\mathbf{D}}_{i,t}]=0$
by the zero-mean property. 
\begin{align}
\sum_{t=0}^{T}E[\ddot{\mathbf{D}}_{i,t}Y_{i,t}]= & \sum_{t=0}^{T}E[\ddot{\mathbf{D}}_{i,t}Y_{i,t}]-\underbrace{\sum_{t=0}^{T}E[\ddot{\mathbf{D}}_{i,t}Y_{i,0}^{\infty}]}_{=0}=\sum_{t=0}^{T}E[\ddot{\mathbf{D}}_{i,t}\left(Y_{i,t}-Y_{i,0}^{\infty}\right)]=\sum_{t=0}^{T}E[\ddot{\mathbf{D}}_{i,t}E[Y_{i,t}-Y_{i,0}^{\infty}\vert E_{i}]]\\
= & \sum_{t=0}^{T}E[\ddot{\mathbf{D}}_{i,t}E[Y_{i,t}-Y_{i,0}^{\infty}\vert E_{i}]]-\sum_{t=0}^{T}\underbrace{E[\ddot{\mathbf{D}}_{i,t}]}_{=0}E[Y_{i,t}^{\infty}-Y_{i,0}^{\infty}]\\
= & \sum_{t=0}^{T}E[\ddot{\mathbf{D}}_{i,t}\left(E[Y_{i,t}-Y_{i,0}^{\infty}\vert E_{i}]-E[Y_{i,t}^{\infty}-Y_{i,0}^{\infty}]\right)].
\end{align}
We can abbreviate the term in the parentheses as $f(E_{i},t)$ to
emphasize it is a function of $E_{i}$ and $t$. Since $E_{i}$ and
$t$ take on discrete values, we can write 
\begin{align}
E[Y_{i,t}-Y_{i,0}^{\infty}\vert E_{i}]-E[Y_{i,t}^{\infty}-Y_{i,0}^{\infty}]\eqqcolon f(E_{i},t) & =\sum_{e=0}^{\infty}f(e,t)\cdot\mathbf{1}\left\{ E_{i}=e\right\} \\
 & =\sum_{e=0}^{\infty}\sum_{\ell=-T}^{T}f(e,e+\ell)\cdot\mathbf{1}\{t-e=\ell\}\cdot\mathbf{1}\left\{ E_{i}=e\right\} \\
 & =\sum_{\ell=-T}^{T}\sum_{e=0}^{\infty}D_{i,t}^{\ell}\mathbf{1}\left\{ E_{i}=e\right\} \cdot f(e,e+\ell).
\end{align}
Even though we sum over $e\in\{0,\dots,\infty\}$, if there are no
never-treated units (or no units in a particular cohort $e'$), then
we have no units take on values $E_{i}=\infty$ (or $E_{i}=e'$).
The range of the summation is therefore still sensible. Note that
we replace $\mathbf{1}\{t-e=\ell\}\cdot\mathbf{1}\left\{ E_{i}=e\right\} $
with $D_{i,t}^{\ell}\cdot\mathbf{1}\left\{ E_{i}=e\right\} $ because
\begin{equation}
D_{i,t}^{\ell}\cdot\mathbf{1}\left\{ E_{i}=e\right\} =\mathbf{1}\{t-E_{i}=\ell\}\cdot\mathbf{1}\left\{ E_{i}=e\right\} =\mathbf{1}\{t-e=\ell\}\cdot\mathbf{1}\left\{ E_{i}=e\right\} .
\end{equation}

Using the above expression, the coefficient $\mu_{g}$ can be written
as
\begin{align}
\mu_{g}= & \mathbf{e}_{g}^{\intercal}\left(\sum_{t=0}^{T}E[\ddot{\mathbf{D}}_{i,t}\mathbf{D}_{i,t}^{\intercal}]\right)^{-1}\sum_{t=0}^{T}\sum_{\ell=-T}^{T}\sum_{e=0}^{\infty}E[\ddot{\mathbf{D}}_{i,t}D_{i,t}^{\ell}\cdot\mathbf{1}\left\{ E_{i}=e\right\} ]f(e,e+\ell)\\
= & \sum_{t=0}^{T}\sum_{l\in g}\sum_{e=0}^{\infty}\underbrace{\mathbf{e}_{g}^{\intercal}\left(\sum_{t=0}^{T}E[\ddot{\mathbf{D}}_{i,t}\mathbf{D}_{i,t}^{\intercal}]\right)^{-1}E[\ddot{\mathbf{D}}_{i,t}D_{i,t}^{\ell}\cdot\mathbf{1}\left\{ E_{i}=e\right\} ]}_{=\omega_{e,\ell}^{g}}f(e,e+\ell)\label{eq:weights-formula}\\
 & +\sum_{t=0}^{T}\sum_{g'\neq g}\sum_{\ell\in g'}\sum_{e=0}^{\infty}\underbrace{\mathbf{e}_{g}^{\intercal}\left(\sum_{t=0}^{T}E[\ddot{\mathbf{D}}_{i,t}\mathbf{D}_{i,t}^{\intercal}]\right)^{-1}E[\ddot{\mathbf{D}}_{i,t}D_{i,t}^{\ell}\cdot\mathbf{1}\left\{ E_{i}=e\right\} ]}_{=\omega_{e,\ell}^{g}}f(e,e+\ell)\\
 & +\sum_{t=0}^{T}\sum_{\ell\in g^{excl}}\sum_{e=0}^{\infty}\underbrace{\mathbf{e}_{g}^{\intercal}\left(\sum_{t=0}^{T}E[\ddot{\mathbf{D}}_{i,t}\mathbf{D}_{i,t}^{\intercal}]\right)^{-1}E[\ddot{\mathbf{D}}_{i,t}D_{i,t}^{\ell}\cdot\mathbf{1}\left\{ E_{i}=e\right\} ]}_{=\omega_{e,\ell}^{g}}f(e,e+\ell).
\end{align}
The superscript $g$ in $\omega_{e,\ell}^{g}$ indexes $\mu_{g}$.
The subscript $\ell$ in $\omega_{e,\ell}^{g}$ indexes $f(e,e+\ell)$.
The expression above the braces makes it clear that the weight $\omega_{e,\ell}^{g}$
is equal to the population regression coefficient on $\mathbf{1}\{t-E_{i}\in g\}$
from regressing $D_{i,t}^{\ell}\cdot\mathbf{1}\left\{ E_{i}=e\right\} $
on all the bin indicators i.e. $\{\mathbf{1}\{t-E_{i}\in g\}\}_{g\in\mathcal{G}}$
and two-way fixed effects included in (\ref{eq:dynamic}). This proves
Proposition \ref{prop: dyno FE coeff weights-1}. We next show several
properties of these weights in the expression of $\mu_{g}$:
\begin{enumerate}
\item For relative times of $\mu_{g}$'s own bin i.e. $\ell\in g$, the
weights sum to one $\sum_{\ell\in g}\sum_{e}\omega_{e,\ell}^{g}=1$. 
\item For relative times belonging to some other bin i.e. $\ell\in g'$
and $g'\neq g$, the weights sum to zero $\sum_{\ell\in g'}\sum_{e}\omega_{e,\ell}^{g}=0$
for each bin $g'$. 
\item For relative times not contained in $\mathcal{G}$ i.e. $\ell\in g^{excl}$,
the weights sum to negative one $\sum_{\ell\in g^{excl}}\sum_{e}\omega_{e,\ell}^{g}=-1$. 
\item If there are never-treated units i.e. $\infty\in supp(E_{i})$, we
have $\omega_{\infty,\ell}^{g}=0$ for all $g$ and $\ell$.
\end{enumerate}
To see that 1) $\sum_{\ell\in g}\sum_{e}\omega_{e,\ell}^{g}=1$, note
that the sum of weights is equal to
\begin{align}
\sum_{t=0}^{T}\sum_{\ell\in g}\sum_{e}\omega_{e,\ell}^{g}=\sum_{t=0}^{T}\sum_{\ell\in g}\sum_{e} & \mathbf{e}_{g}^{\intercal}\left(\sum_{t=0}^{T}E[\ddot{\mathbf{D}}_{i,t}\mathbf{D}_{i,t}^{\intercal}]\right)^{-1}E[\ddot{\mathbf{D}}_{i,t}D_{i,t}^{\ell}\cdot\mathbf{1}\left\{ E_{i}=e\right\} ]\\
= & \mathbf{e}_{g}^{\intercal}\left(\sum_{t=0}^{T}E[\ddot{\mathbf{D}}_{i,t}\mathbf{D}_{i,t}^{\intercal}]\right)^{-1}\sum_{t=0}^{T}E[\ddot{\mathbf{D}}_{i,t}\sum_{\ell\in g}D_{i,t}^{\ell}\cdot\sum_{e}\mathbf{1}\left\{ E_{i}=e\right\} ]\\
= & \mathbf{e}_{g}^{\intercal}\left(\sum_{t=0}^{T}E[\ddot{\mathbf{D}}_{i,t}\mathbf{D}_{i,t}^{\intercal}]\right)^{-1}\sum_{t=0}^{T}E[\ddot{\mathbf{D}}_{i,t}\mathbf{1}\{t-E_{i}\in g\}].
\end{align}
It is thus the population regression coefficient on $\mathbf{1}\{t-E_{i}\in g\}$
from regressing $\mathbf{1}\{t-E_{i}\in g\}$ on $\mathbf{D}_{i,t}$
and the unit and time fixed effects, which is just one. Similarly,
for each $g'\neq g$, the sum of weights is the population regression
coefficient on $\mathbf{1}\{t-E_{i}\in g'\}$ from regressing $\mathbf{1}\{t-E_{i}\in g\}$
on $\mathbf{D}_{i,t}$ and the unit and time fixed effects, which
is zero. To see that the weights from all excluded relative time add
up to negative one across cohorts, note that the sum of these weights
is equal to
\begin{align}
\sum_{\ell\in g^{excl}}\sum_{e}\omega_{e,\ell}^{g}= & \mathbf{e}_{g}^{\intercal}\sum_{\ell\in g^{excl}}\left(\sum_{t=0}^{T}E[\ddot{\mathbf{D}}_{i,t}\mathbf{D}_{i,t}^{\intercal}]\right)^{-1}\sum_{t=0}^{T}E[\ddot{\mathbf{D}}_{i,t}D_{i,t}^{\ell}]\\
= & \mathbf{e}_{g}^{\intercal}\left(\sum_{t=0}^{T}E[\ddot{\mathbf{D}}_{i,t}\mathbf{D}_{i,t}^{\intercal}]\right)^{-1}\sum_{t=0}^{T}E\Bigl[\ddot{\mathbf{D}}_{i,t}\sum_{\ell\in g^{excl}}D_{i,t}^{\ell}\Bigr]\\
= & \mathbf{e}_{g}^{\intercal}\left(\sum_{t=0}^{T}E[\ddot{\mathbf{D}}_{i,t}\mathbf{D}_{i,t}^{\intercal}]\right)^{-1}\sum_{t=0}^{T}E[\ddot{\mathbf{D}}_{i,t}\bigl(1-\sum_{g\in\mathcal{G}}\sum_{\ell\in g}D_{i,t}^{\ell}\bigr)]\\
= & -\mathbf{e}_{g}^{\intercal}\left(\sum_{t=0}^{T}E[\ddot{\mathbf{D}}_{i,t}\mathbf{D}_{i,t}^{\intercal}]\right)^{-1}\sum_{t=0}^{T}E[\ddot{\mathbf{D}}_{i,t}\sum_{g\in\mathcal{G}}\sum_{\ell\in g}D_{i,t}^{\ell}]\\
= & -\mathbf{e}_{g}^{\intercal}\sum_{g\in\mathcal{G}}\left(\sum_{t=0}^{T}E[\ddot{\mathbf{D}}_{i,t}\mathbf{D}_{i,t}^{\intercal}]\right)^{-1}\sum_{t=0}^{T}E[\ddot{\mathbf{D}}_{i,t}\mathbf{1}\{t-E_{i}\in g\}]=-1
\end{align}
where the third equality follows from $\sum_{\ell\in g^{excl}}D_{i,t}^{\ell}+\sum_{g\in\mathcal{G}}\sum_{\ell\in g}D_{i,t}^{\ell}=\sum_{-T\leq\ell\leq T}D_{i,t}^{\ell}=1$.
The fourth equality follows from $\sum_{t}E[\ddot{\mathbf{D}}_{i,t}]=0$
due to the zero mean property of double-demeaning proved in Lemma
\ref{lem:Properties-of-double-demea}. The last line simplifies to
negative one because each term in the summation 
\begin{equation}
\left(\sum_{t=0}^{T}E[\ddot{\mathbf{D}}_{i,t}\mathbf{D}_{i,t}^{\intercal}]\right)^{-1}\sum_{t=0}^{T}E[\ddot{\mathbf{D}}_{i,t}\mathbf{1}\{t-E_{i}\in g\}]\label{eq:dyno-excl-coeff}
\end{equation}
equals the population regression coefficient on $\mathbf{D}_{i,t}$
from regressing $\mathbf{1}\{t-E_{i}\in g\}$ on $\mathbf{D}_{i,t}$
and the unit and time fixed effects, which is a column vector with
one in the entry corresponding to $\mathbf{1}\{t-E_{i}\in g\}$. Summing
over $g$, the above expression is equal to a column vector of ones.

Finally, we note that the weight for the never-treated cohort is always
zero $\omega_{\infty,\ell}^{g}=0$. This is because $D_{i,t}^{\ell}=0$
for all $\ell$ when $E_{i}=\infty$.
\end{proof}

\subsection*{Proof of Proposition \ref{prop: dyno FE coeff weights}}
\begin{proof}
Under the parallel trends assumption, we can replace each term in
Proposition \ref{prop: dyno FE coeff weights-1} with 
\begin{equation}
E[Y_{i,e+\ell}-Y_{i,0}^{\infty}\vert E_{i}]-E[Y_{i,e+\ell}^{\infty}-Y_{i,0}^{\infty}]=CATT_{e,\ell}+\underbrace{E[Y_{i,t}^{\infty}-Y_{i,0}^{\infty}\vert E_{i}]-E[Y_{i,t}^{\infty}-Y_{i,0}^{\infty}]}_{=0}
\end{equation}
 for $t=e+\ell$. 
\end{proof}

\subsection*{Proof of Proposition \ref{prop:dyno FE coeff weights no anti}}
\begin{proof}
Under the no anticipation assumption, setting the corresponding terms
to zero in the linear combination underlying $\mu_{\underline{g}}$
displayed in Proposition \ref{prop: dyno FE coeff weights}.
\end{proof}

\subsection*{Proof of Proposition \ref{prop:dyno FE coeff weights homo}}
\begin{proof}
Under the assumption of homogeneous treatment effect, we have $CATT_{e,\ell}=ATT_{\ell}$
for all $e$. This simplifies Proposition \ref{prop: dyno FE coeff weights}
to 
\begin{align}
\sum_{\ell\in g}\left(\sum_{e}\omega_{e,\ell}^{g}\right)ATT_{\ell}+\sum_{g'\neq g}\sum_{\ell\in g'}\left(\sum_{e}\omega_{e,\ell}^{g}\right)ATT_{\ell}+\sum_{\ell\in g^{excl}}\left(\sum_{e}\omega_{e,\ell}^{g}\right)ATT_{\ell}.
\end{align}
Recall the weight $\omega_{e,\ell}^{g}$ is equal to the population
regression coefficient on $\mathbf{1}\{t-E_{i}\in g\}$ from regressing
$D_{i,t}^{\ell}\cdot\mathbf{1}\left\{ E_{i}=e\right\} $ on all the
bin indicators, i.e. $\{\mathbf{1}\{t-E_{i}\in g\}\}_{g\in\mathcal{G}}$
and two-way fixed effects included in (\ref{eq:dynamic}). Summing
over $e$ would imply the weight $\omega_{\ell}^{g}\coloneqq\sum_{e}\omega_{e,\ell}^{g}$
is equal to the population regression coefficient on $\mathbf{1}\{t-E_{i}\in g\}$
from regressing $D_{i,t}^{\ell}$ on all the bin indicators, i.e.
$\{\mathbf{1}\{t-E_{i}\in g\}\}_{g\in\mathcal{G}}$ and two-way fixed
effects included in (\ref{eq:dynamic}).
\end{proof}

\section{Supplementary results related to IW estimator}

In this section, we first provide detailed definition for the IW estimator,
and then state its asymptotic behavior. We then prove its asymptotic
normality and derive its asymptotic variance in Section \ref{subsec:IW-proof}. 
\begin{defn}
\label{def:The-IW-estimator}The IW estimator for $\nu_{g}$, a weighted
average of $CATT_{e,\ell\in g}$ is constructed via the following
three steps. We focus on the setting without never-treated units,
using pre-period $s=-1$ and the last-treated cohort as the control
cohort $C=\{\max\{E_{i}\}\}$. The setting with never-treated units
can be defined similarly by modifying the interacted specification.
\begin{description}
\item [{Step~1}] Estimate the $CATT_{e,\ell}$ using the interacted specification.
\begin{align}
Y_{i,t} & =\alpha_{i}+\lambda_{t}+\sum_{e\neq\max\{E_{i}\}}\sum_{\ell\neq-1}\delta_{e,\ell}(\mathbf{1}\{E_{i}=e\}\cdot D_{i,t}^{\ell})+\epsilon_{i,t}\\
 & =\alpha_{i}+\lambda_{t}+\mathbf{B}_{i,t}^{\intercal}\boldsymbol{\delta}+\epsilon_{i,t}
\end{align}
on observations from $t=0,\dots,\max\{E_{i}\}-1$ and $E_{i}\neq0$.
We need to drop time period beyond $\max\{E_{i}\}$ because DID estimators
for $CATT_{e,\ell}$ for $e+\ell\geq\max\{E_{i}\}$do not exist as
explained in the main text. We need to exclude cohort $0$ from estimation
because DID estimators for $CATT_{0,\ell}$ do not exist as explained
in the main text. Note that among regressors we exclude interactions
with $\mathbf{1}\{E_{i}=\max\{E_{i}\}\}$ and we exclude interactions
with $D_{i,t}^{-1}$. Here $\mathbf{B}_{i,t}$ is a column vector
collecting the interactions $\mathbf{1}\{E_{i}=e\}\cdot D_{i,t}^{\ell}$.
Similarly, $\boldsymbol{\delta}$ is a column vector collecting the
coefficients $\delta_{e,\ell}$ on $\mathbf{1}\{E_{i}=e\}\cdot D_{i,t}^{\ell}$.
The matrix notation is used later to derive the asymptotic variance
of IW estimators.\\
\item [{Step~2}] Estimate the weights, which are cohort shares among cohorts
that experience at least $\ell$ periods of treatment relative to
the initial treatment. \\
Denote by $N_{e}\coloneqq\sum_{i=1}^{N}\mathbf{1}\left\{ E_{i}=e\right\} $
the number of units in cohort $e$. Denote by $h^{\ell}=\{e:1-\ell\leq e\leq\max\{E_{i}\}-1-\ell\}$
to be the set of cohorts that experience at least $\ell$ periods
of treatment relative to the initial treatment. Below $vec\left(A\right)$
vectorizes matrix $A$ by stacking its columns. 

Define $\widehat{\mathbf{f}}^{\ell}$ to be a matrix with its $\left(t,e\right)^{th}$
entry equal to $\mathbf{1}\left\{ t-e=\ell\right\} \cdot N_{e}/\sum_{e\in h^{\ell}}N_{e}$.
Here $\mathbf{1}\left\{ t-e=\ell\right\} $ indicates when cohort
$e$ experiences exactly $\ell$ periods of treatment and $N_{e}/\sum_{e\in h^{\ell}}N_{e}$
is equal to the sample share of units in cohort $e$ among units that
experience at least $l$ periods of treatment. Denote by $\mathbf{f}^{\ell}$
the probability limit of $\widehat{\mathbf{f}}^{\ell}$, which is
a matrix with its $\left(t,e\right)^{th}$ entry equal to $\mathbf{1}\left\{ t-e=\ell\right\} \cdot Pr\{E_{i}=t-l\mid E_{i}\in h^{\ell}\}$.
For example, with $E_{i}\in\{0,1,2,3\}$ for $T=3$ and $\ell=0$,
we have $h^{0}=\{1,2\}$ and thus
\begin{equation}
\widehat{\mathbf{f}}^{0}=\left(\begin{array}{cc}
\frac{N_{1}}{N_{1}+N_{2}} & 0\\
0 & \frac{N_{2}}{N_{1}+N_{2}}
\end{array}\right)
\end{equation}
and its probability limit is 
\begin{equation}
\mathbf{f}^{0}=\left(\begin{array}{cc}
Pr\left\{ E_{i}=1\vert1\le E_{i}\leq2\right\}  & 0\\
0 & Pr\left\{ E_{i}=2\vert1\le E_{i}\leq2\right\} 
\end{array}\right).
\end{equation}
In proof below, we show that the weight matrix estimator $\widehat{\mathbf{f}}^{\ell}$
is asymptotically normal $\sqrt{N}(vec(\widehat{\mathbf{f}}^{\ell})-vec(\mathbf{f}^{\ell}))\rightarrow_{d}N(0,\Sigma_{f^{\ell}})$. 
\item [{Step~3}] Compute the IW estimator as the weighted sum of $\widehat{\delta}_{e,\ell}$
(estimated in Step 1) using weights (estimated in Step 2).

To form an estimator alternative to the dynamic FE estimator $\widehat{\mu}_{g}$
from Regression (\ref{eq:dynamic}), we can use 
\begin{equation}
\widehat{\nu}_{g}\coloneqq\frac{1}{\left|g\right|}\sum_{\ell\in g}\sum_{e\in h^{\ell}}\frac{N_{e}}{\sum_{e\in h^{\ell}}N_{e}}\widehat{\delta}{}_{e,\ell}=\frac{1}{\left|g\right|}\sum_{\ell\in g}vec(\widehat{\mathbf{f}}^{\ell})^{\intercal}\widehat{\boldsymbol{\delta}}.
\end{equation}

\end{description}
\end{defn}
With a few standard assumptions (which we present together below as
Assumption \ref{assu:(The-saturated-regression}) on Regression (\ref{eq: saturated event studies model}),
we can show that the IW estimators are asymptotically normal. 
\begin{assumption}
\emph{\label{assu:(The-saturated-regression}(The saturated regression
assumptions).}
\begin{enumerate}
\item There are observations from at least two cohorts that are not treated
in $t=0$.
\item Independent, identically distributed cross-sectional observations:
$\{(E_{i},\mathbf{Y}_{i}):i=1,2,\dots,N\}$ are i.i.d. draws from
their joint distribution where $\mathbf{Y}_{i}=(Y_{i,0},\dots,Y_{i,T})^{\intercal}$
is a $T\times1$ vector.
\item Large outliers are unlikely: $(\mathbf{B}_{i,t},\epsilon_{i,t})$
have nonzero finite fourth moments.
\item Denote by $\ddot{\mathbf{B}}$ the data matrix, whose rows consist
of $\ddot{\mathbf{B}}_{i,t}^{\intercal}$, double-demeaned version
of $\mathbf{B}_{i,t}^{\intercal}$. Assume $\ddot{\mathbf{B}}$ has
full rank. If $\ddot{\mathbf{B}}$ is reduced-rank because cohort
$e$ is empty, then discard regressors involving $\mathbf{1}\{E_{i}=e\}$. 
\end{enumerate}
\end{assumption}
Denote by $\boldsymbol{\delta}$ the probability limit of $\widehat{\boldsymbol{\delta}}$,
which is a vector of $CATT_{e,\ell}$. We next state the asymptotic
distribution of the IW estimators specifically for $\nu_{\ell}$,
the average effect at relative time $\ell$. Results for the more
general case of $\nu_{g}$, which is the average effects across relative
times $\ell\in g$, can be derived similarly. Note that we use a clustered
variance-covariance structure to allow the possibility that $Y_{i,t}$
are dependent across $t$ due to serial correlation. 
\begin{prop}
\emph{\label{prop: CAN IW dynamic}(Consistency and asymptotic normality
of the }IW\emph{ estimators for $\nu_{\ell}$).} Under the assumptions
of Proposition \ref{prop:The-DID-estimator} and Assumption \ref{assu:(The-saturated-regression},
the IW estimator converges in probability to 
\begin{equation}
\widehat{\nu}_{\ell}\rightarrow_{p}\sum_{e\in h^{\ell}}Pr\left\{ E_{i}=e\mid E_{i}\in h^{\ell}\right\} CATT_{e,\ell}=vec(\mathbf{f}^{\ell})^{\intercal}\boldsymbol{\delta}.
\end{equation}
The asymptotic distribution of this estimator is 
\begin{equation}
\sqrt{N}\left(\widehat{\nu}_{\ell}-vec\left(\mathbf{f}^{\ell}\right)^{\intercal}\boldsymbol{\delta}\right)\rightarrow_{d}N\left(0,\boldsymbol{\delta}^{\intercal}\Sigma_{f^{\ell}}\mathbf{\boldsymbol{\delta}}+\Sigma_{\ell}\right)
\end{equation}
for $\Sigma_{f^{\ell}}$ the asymptotic variance of $\sqrt{N}(vec(\widehat{\mathbf{f}}^{\ell})-vec(\mathbf{f}^{\ell}))$
where $\widehat{\mathbf{f}}^{\ell}$ is the weight matrix estimator
and
\begin{alignat}{2}
\begin{alignedat}[t]{1}\boldsymbol{V}_{\ddot{\mathbf{B}}} & =\sum_{t=0}^{\max\{E_{i}\}-1}E[\ddot{\mathbf{B}}_{i,t}\ddot{\mathbf{B}}_{i,t}^{\intercal}]\end{alignedat}
 & \:\ \ \ \  & \begin{alignedat}[t]{1}\Sigma_{\ell} & =vec(\mathbf{f}^{\ell})^{\intercal}\boldsymbol{V}_{\ddot{\mathbf{B}}}^{-1}Var\left(\sum_{t=0}^{\max\{E_{i}\}-1}\ddot{\mathbf{B}}_{i,t}\ddot{\epsilon}_{i,t}\right)\boldsymbol{V}_{\ddot{\mathbf{B}}}^{-1}vec(\mathbf{f}^{\ell})\end{alignedat}
.
\end{alignat}
\end{prop}

\subsection{Proof of the validity of the IW estimator\label{subsec:IW-proof}}

\subsection*{Proof of Proposition \ref{prop:The-DID-estimator}}
\begin{proof}
Provided that the DID estimator is well-defined with pre-period $s$
and control cohorts $C$, we first show the DID estimator is an unbiased
and consistent estimator for $E[Y_{i,e+\ell}-Y_{i,s}\mid E_{i}=e]-E[Y_{i,e+\ell}-Y_{i,s}\mid E_{i}\in C]$.
We prove the first term in the DID estimator $\frac{\mathbb{E}_{N}[\left(Y_{i,e+\ell}-Y_{i,s}\right)\cdot\text{\ensuremath{\mathbf{1}}}\left\{ E_{i}=e\right\} ]}{\mathbb{E}{}_{N}[\text{\ensuremath{\mathbf{1}}}\left\{ E_{i}=e\right\} ]}$
is unbiased and consistent for $E[Y_{i,e+\ell}-Y_{i,s}\mid E_{i}=e$.
The argument for the second term in the DID estimator follows similarly. 

For unbiasedness, note that by the Law of Iterated Expectations and
linearity of $\mathbb{E}_{N}$, we have
\begin{align}
E[\frac{\mathbb{E}_{N}[\left(Y_{i,e+\ell}-Y_{i,s}\right)\cdot\text{\ensuremath{\mathbf{1}}}\left\{ E_{i}=e\right\} ]}{\mathbb{E}{}_{N}[\text{\ensuremath{\mathbf{1}}}\left\{ E_{i}=e\right\} ]}] & =E\biggl[E\Bigl[\frac{\mathbb{E}_{N}[\left(Y_{i,e+\ell}-Y_{i,s}\right)\cdot\text{\ensuremath{\mathbf{1}}}\left\{ E_{i}=e\right\} ]}{\mathbb{E}{}_{N}[\text{\ensuremath{\mathbf{1}}}\left\{ E_{i}=e\right\} ]}\mid E_{i}\Bigr]\biggr]\\
 & =E\biggl[\frac{\mathbb{E}_{N}\Bigl[E[\left(Y_{i,e+\ell}-Y_{i,s}\right)\cdot\text{\ensuremath{\mathbf{1}}}\left\{ E_{i}=e\right\} \mid E_{i}]\Bigr]}{\mathbb{E}{}_{N}[\text{\ensuremath{\mathbf{1}}}\left\{ E_{i}=e\right\} ]}\biggr]\\
 & =E\biggl[\frac{\mathbb{E}_{N}\Bigl[E[Y_{i,e+\ell}-Y_{i,s}\mid E_{i}=e]\cdot\text{\ensuremath{\mathbf{1}}}\left\{ E_{i}=e\right\} \Bigr]}{\mathbb{E}{}_{N}[\text{\ensuremath{\mathbf{1}}}\left\{ E_{i}=e\right\} ]}\biggr]\\
 & =E[\frac{E[\left(Y_{i,e+\ell}-Y_{i,s}\right)\mid E_{i}=e]\cdot\mathbb{E}_{N}[\text{\ensuremath{\mathbf{1}}}\left\{ E_{i}=e\right\} ]}{\mathbb{E}{}_{N}[\text{\ensuremath{\mathbf{1}}}\left\{ E_{i}=e\right\} ]}]\\
 & =E[\left(Y_{i,e+\ell}-Y_{i,s}\right)\mid E_{i}=e].
\end{align}
For consistency, by the Law of Large Numbers the numerator and the
denominator converge in probability to $E[\left(Y_{i,e+\ell}-Y_{i,s}\right)\cdot\text{\ensuremath{\mathbf{1}}}\left\{ E_{i}=e\right\} ]$
and $Pr\left\{ E_{i}=e\right\} $ respectively. By the Law of Iterated
Expectations and Slutsky's theorem, it converges in probability to
$E[Y_{i,e+\ell}-Y_{i,s}\mid E_{i}=e]$. 

To show that the DID estimator is an unbiased and consistent estimator
for $CATT_{e,\ell}$, it remains to show $E[Y_{i,e+\ell}-Y_{i,s}\mid E_{i}=e]-E[Y_{i,e+\ell}-Y_{i,s}\mid E_{i}\in C]=CATT_{e,\ell}$. 

Since $s<e$ and $c>e+\ell$, we have 
\begin{align}
 & E[Y_{i,e+\ell}-Y_{i,s}\mid E_{i}=e]-E[Y_{i,e+\ell}-Y_{i,s}\mid E_{i}\in C]\\
= & E[Y_{i,e+\ell}^{e}-Y_{i,s}^{e}\mid E_{i}=e]-\sum_{c\in C}Pr\left\{ E_{i}=c\mid E_{i}\in C\right\} E[Y_{i,e+\ell}^{c}-Y_{i,s}^{c}\mid E_{i}=c]\\
= & E[Y_{i,e+\ell}^{e}-Y_{i,s}^{\infty}\mid E_{i}=e]-\sum_{c\in C}Pr\left\{ E_{i}=c\mid E_{i}\in C\right\} E[Y_{i,e+\ell}^{\infty}-Y_{i,s}^{\infty}\mid E_{i}=c]\\
= & E[Y_{i,e+\ell}^{e}-Y_{i,e+\ell}^{\infty}\mid E_{i}=e]+E[Y_{i,e+\ell}^{\infty}-Y_{i,s}^{\infty}\mid E_{i}=e]\\
 & -\sum_{c\in C}Pr\left\{ E_{i}=c\mid E_{i}\in C\right\} E[Y_{i,e+\ell}^{\infty}-Y_{i,s}^{\infty}\mid E_{i}=c]\\
= & E[Y_{i,e+\ell}^{e}-Y_{i,e+\ell}^{\infty}\vert E_{i}=e]+E[Y_{i,e+\ell}^{\infty}-Y_{i,s}^{\infty}]-E[Y_{i,e+\ell}^{\infty}-Y_{i,s}^{\infty}]\\
= & E[Y_{i,e+\ell}^{e}-Y_{i,e+\ell}^{\infty}\vert E_{i}=e]
\end{align}
where the second equality follows from Assumption \ref{assum:no-anti}
and the fourth equality follows from Assumption \ref{assum:PT-baseline}. 
\end{proof}

\subsection*{Proof of Proposition \ref{prop: CAN IW dynamic} }
\begin{proof}
We first show the consistency and asymptotic normality of the weight
estimators. The weight estimators $\widehat{\mathbf{f}}^{\ell}$ are
consistent since $\frac{N_{e}}{\sum_{e\in h^{\ell}}N_{e}}\rightarrow_{p}\frac{Pr\left\{ E_{i}=e\right\} }{Pr\left\{ E_{i}\in h^{\ell}\right\} }=Pr\left\{ E_{i}=e\mid E_{i}\in h^{\ell}\right\} $
by the Law of Large Numbers and Slutsky's theorem. Note that $\frac{N_{e}}{\sum_{e\in h^{\ell}}N_{e}}$
is also the regression coefficient estimator from the following cross-sectional
regression 
\begin{equation}
\mathbf{1}\left\{ E_{i}=e\right\} =\beta\mathbf{1}\left\{ E_{i}\in h^{\ell}\right\} +\eta_{i}\left(e\right)
\end{equation}
with population regression coefficient equal to $\beta=Pr\left\{ E_{i}=e\mid E_{i}\in h^{\ell}\right\} $.
Then by OLS asymptotics which holds as $E_{i}$ are iid by assumption
and $\eta_{i}\left(e\right)$ is bounded, we have 
\begin{equation}
\sqrt{N}\left(\frac{N_{e}}{\sum_{e\in h^{\ell}}N_{e}}-Pr\left\{ E_{i}=e\mid E_{i}\in h^{\ell}\right\} \right)\rightarrow_{d}N\left(0,\frac{E[\mathbf{1}\left\{ E_{i}\in h^{\ell}\right\} ^{2}\eta_{i}^{2}\left(e\right)]}{E[\mathbf{1}\left\{ E_{i}\in h^{\ell}\right\} ^{2}]^{2}}\right).
\end{equation}
Note that $\mathbf{1}\left\{ E_{i}\in h^{\ell}\right\} ^{2}=\mathbf{1}\left\{ E_{i}\in h^{\ell}\right\} $
so the asymptotic variance is equal to
\begin{align}
 & \frac{E[\eta_{i}^{2}\left(e\right)\mid E_{i}\in h^{\ell}]Pr\left\{ E_{i}\in h^{\ell}\right\} }{Pr\left\{ E_{i}\in h^{\ell}\right\} ^{2}}=\frac{E[\eta_{i}^{2}\left(e\right)\mid E_{i}\in h^{\ell}]}{Pr\left\{ E_{i}\in h^{\ell}\right\} }.
\end{align}

Similarly, for a pair of cohorts with $e\neq e'$, $\frac{N_{e}}{\sum_{e\in h^{\ell}}N_{e}}$
and $\frac{N_{e'}}{\sum_{e\in h^{\ell}}N_{e}}$ are asymptotically
correlated with covariance $E[\eta_{i}\left(e\right)\eta_{i}\left(e'\right)\mid E_{i}\in h^{\ell}]/Pr\left\{ E_{i}\in h^{\ell}\right\} $.
Thus, $vec\left(\widehat{\mathbf{f}}^{\ell}\right)$ has asymptotic
distribution 
\begin{equation}
\sqrt{N}\left(vec\left(\widehat{\mathbf{f}}^{\ell}\right)-vec\left(\mathbf{f}^{\ell}\right)\right)\rightarrow_{d}N\left(0,\Sigma_{f^{\ell}}\right).
\end{equation}
Here  $\Sigma_{f^{\ell}}$ is a matrix with diagonal entries equal
to $\frac{E[\eta_{i}^{2}\left(e\right)\mid E_{i}\in h^{\ell}]}{Pr\left\{ E_{i}\in h^{\ell}\right\} },$
and off-diagonal entries equal to $\frac{E[\eta_{i}\left(e\right)\eta_{i}\left(e'\right)\mid E_{i}\in h^{\ell}]}{Pr\left\{ E_{i}\in h^{\ell}\right\} }.$

By consistency of the weight estimators and Proposition \ref{prop:The-DID-estimator},
we prove the consistency of the IW estimator, the first part of Proposition
\ref{prop: CAN IW dynamic}. 

We next show the asymptotic normality of the $\widehat{\delta}_{e,\ell}$.
The standard OLS asymptotics applies because by assumption after double
demeaning, the data $(\ddot{\mathbf{B}}_{i,t},\ddot{\epsilon}_{i,t})$
is iid across $i$ and has nonzero finite fourth moments. The asymptotic
distribution of this estimator is thus
\begin{equation}
\sqrt{N}\left(\widehat{\boldsymbol{\delta}}-\boldsymbol{\delta}\right)\rightarrow_{d}N\left(0,\boldsymbol{V}_{\ddot{\mathbf{B}}}^{-1}Var\left(\sum_{t=0}^{T-1}\ddot{\mathbf{B}}_{i,t}\ddot{\epsilon}_{i,t}\right)\boldsymbol{V}_{\ddot{\mathbf{B}}}^{-1}\right)
\end{equation}
where $\begin{alignedat}[t]{1}\boldsymbol{V}_{\ddot{\mathbf{B}}} & =\sum_{t=0}^{T-1}E[\ddot{\mathbf{B}}_{i,t}\ddot{\mathbf{B}}_{i,t}^{\intercal}]\end{alignedat}
$.

Lastly, by the delta method, we have 
\begin{equation}
\sqrt{N}\left(vec\left(\widehat{\mathbf{f}}^{\ell}\right)^{\intercal}\widehat{\boldsymbol{\delta}}-vec\left(\mathbf{f}^{\ell}\right)^{\intercal}\boldsymbol{\delta}\right)\rightarrow_{d}N\left(0,\boldsymbol{\delta}^{\intercal}\Sigma_{f^{\ell}}\mathbf{\boldsymbol{\delta}}+\Sigma_{\ell}\right)
\end{equation}
where $\begin{alignedat}[t]{1}\Sigma_{\ell} & =vec(\mathbf{f}^{\ell})^{\intercal}\boldsymbol{V}_{\ddot{\mathbf{B}}}^{-1}Var\left(\sum_{t=0}^{T-1}\ddot{\mathbf{B}}_{i,t}\ddot{\epsilon}_{i,t}^{2}\ddot{\mathbf{B}}_{i,t}^{\intercal}\right)\boldsymbol{V}_{\ddot{\mathbf{B}}}^{-1}vec(\mathbf{f}^{\ell})\end{alignedat}
$. This follows because $vec\left(\widehat{\mathbf{f}}^{\ell}\right)$
and $\widehat{\boldsymbol{\delta}}$ are uncorrelated: the asymptotic
covariance between $\frac{N_{e}}{\sum_{e\in h^{\ell}}N_{e}}$ and
$\widehat{\boldsymbol{\delta}}$ is equal to 
\begin{equation}
\frac{\boldsymbol{V}_{\ddot{\mathbf{B}}}^{-1}Cov(\mathbf{1}\left\{ E_{i}\in h^{\ell}\right\} \eta_{i}\left(e\right),\ddot{\mathbf{B}}_{i,t}\ddot{\epsilon}_{i,t})}{E[\mathbf{1}\left\{ E_{i}\in h^{\ell}\right\} ^{2}]^{2}}=\frac{\boldsymbol{V}_{\ddot{\mathbf{B}}}^{-1}E[\ddot{\mathbf{B}}\eta_{i}\left(e\right)\ddot{\epsilon}_{i,t}\mid E_{i}\in h^{\ell}]}{E[\mathbf{1}\left\{ E_{i}\in h^{\ell}\right\} ^{2}]^{2}}.
\end{equation}
Since $\ddot{\mathbf{B}}$ and $\eta_{i}\left(e\right)$ are functions
of $E_{i}$, we have $E[\ddot{\mathbf{B}}\eta_{i}\left(e\right)\ddot{\epsilon}_{i,t}\mid E_{i}\in h^{\ell}]=E[\ddot{\mathbf{B}}\eta_{i}\left(e\right)E[\ddot{\epsilon}_{i,t}\mid E_{i},E_{i}\in h^{\ell}]]$.
Furthermore, specification (\ref{eq: saturated event studies model})
is saturated in $E_{i}$ and relative time so $E[\ddot{\epsilon}_{i,t}\mid E_{i}]=0$.
This proves the asymptotic asymptotic normality of the IW estimators,
the second part of Proposition \ref{prop: CAN IW dynamic}.
\end{proof}